\documentclass[journal]{IEEEtran}
\IEEEoverridecommandlockouts
\usepackage{booktabs}
\usepackage{threeparttable}
\usepackage{amsmath,graphicx,cite,amssymb}
\usepackage{amsfonts,multirow,bm,array,setspace}
\usepackage{graphicx,float,cite,amssymb}
\usepackage{multirow,bm,array,setspace}
\usepackage{textcomp}
\usepackage{psfrag}
\usepackage{color}
\usepackage{pstricks}
\usepackage{bbm}
\usepackage{algorithm, algorithmic}
\usepackage{cuted}
\usepackage[caption=false,font=footnotesize]{subfig}
\usepackage{xpatch}
\usepackage{stfloats}

\newcommand{\bH}{\bm{H}}

\newcommand{\bS}{\bm{S}}
\newcommand{\bs}{\bm{s}}
\newcommand{\bI}{\bm{I}}

\newcommand{\bV}{\bm V}
\newcommand{\bu}{\bm u}

\newcommand{\by}{\bm{y}}

\newcommand{\bZ}{\bm{Z}}
\newcommand{\bQ}{\bm{Q}}

\newcommand{\bA}{\bm{A}}
\newcommand{\ba}{\bm{a}}
\newcommand{\bB}{\bm{B}}
\newcommand{\bG}{\bm{G}}
\newcommand{\bM}{\bm{M}}
\newcommand{\bF}{\bm{F}}
\newcommand{\bX}{\bm{X}}

\newcommand{\hh}{\mathrm{H}}

\newcommand{\tr}{\operatorname{tr}}

\newcommand{\bn}{\bm{n}}
\newcommand{\bE}{\bm{E}}
\newcommand{\bP}{\bm{P}}
\newcommand{\be}{\bm{e}}
\newcommand{\vecnotation}[1]{[#1]} 

\newtheorem{proposition}{Proposition}

\def\BibTeX{{\rm B\kern-.05em{\sc i\kern-.025em b}\kern-.08em
    T\kern-.1667em\lower.7ex\hbox{E}\kern-.125emX}}

\makeatletter
\newcommand{\distas}[1]{\mathbin{\overset{#1}{\kern\z@\sim}}}%
\newsavebox{\mybox}\newsavebox{\mysim}
\newcommand{\distras}[1]{%
	\savebox{\mybox}{\hbox{\kern1pt$\scriptstyle#1$\kern1pt}}%
	\savebox{\mysim}{\hbox{$\sim$}}%
	\mathbin{\overset{#1}{\kern\z@\resizebox{\wd\mybox}{\ht\mysim}{$\sim$}}}%
}

\ExplSyntaxOn
\cs_new:Npn \bibColoredItems #1#2
  {
    \clist_map_inline:nn {#2} { \cs_new:cpn {bib@colored@##1} {#1} } 
  }
\ExplSyntaxOff

\newcommand\bib@setcolor[1]{%
  \ifcsname bib@colored@#1\endcsname
    \expandafter\color\expandafter{\csname bib@colored@#1\endcsname}
  \else
    \normalcolor
  \fi
}

\xpatchcmd\@bibitem
  {\item}
  {\bib@setcolor{#1}\item}
  {}{\PatchFailed}

\xpatchcmd\@lbibitem
  {\item}
  {\bib@setcolor{#2}\item}
  {}{\PatchFailed}

\begin{document}

\title{Ellipsoidal Manifold Optimization for Distributed Antenna Beamforming}

\author{
        \IEEEauthorblockN{
        Minhao Zhu, \IEEEmembership{Graduate Student Member,~IEEE}, Kaiming Shen, \IEEEmembership{Senior Member,~IEEE}
        } 
        
        \thanks{
            Manuscript submitted on \today. 

            The authors are with the Future Network of Intelligence Institute (FNii), and the School of Science and Engineering at The Chinese University of Hong Kong, Shenzhen, China (e-mail: minhaozhu@link.cuhk.edu.cn; shenkaiming@cuhk.edu.cn). 
            }
}

\maketitle

\begin{abstract}
This paper addresses the weighted sum-rate (WSR) maximization problem in a downlink distributed antenna system subject to per-cluster power constraints. This optimization scenario presents significant challenges due to the high dimensionality of beamforming variables in dense antenna deployments and the structural complexity of multiple independent power constraints. To overcome these difficulties, we generalize the low-dimensional subspace property—previously established for sum-power constraints—to the per-cluster power constraint case. We prove that all stationary-point beamformers reside in a reduced subspace spanned by the channel vectors of the corresponding antenna cluster. Leveraging this property, we reformulate the original high-dimensional constrained problem into an unconstrained optimization task over a product of ellipsoidal manifolds, thereby achieving significant dimensionality reduction. We systematically derive the necessary Riemannian geometric structures for this specific manifold, including the tangent space, Riemannian metric, orthogonal projection, retraction, and vector transport. Subsequently, we develop a tailored Riemannian conjugate gradient algorithm to solve the reformulated problem. Numerical simulations demonstrate that the proposed algorithm achieves the same local optima as standard benchmarks, such as the weighted minimum mean square error (WMMSE) method and conventional manifold optimization, but with substantially higher computational efficiency and scalability, particularly as the number of antenna clusters increases.
\end{abstract}


\section{Introduction}

The distributed antenna system has emerged as a promising architecture for next-generation wireless networks, where transmit antennas are placed at different positions throughout the cellular area, as opposed to the conventional transmit antennas that are co-located at the base station (BS) \cite{kerpez1996distributed,lin2014distributed}. The main advantage of distributed antennas is to reduce the distance between user terminals and antennas, thus enhancing the cellular coverage and the received signal power. From an optimization perspective, the beamforming problem for the distributed antenna system does not seem much different much from the conventional, except that the power constraint is now imposed on each antenna cluster---which refers to the set of antennas placed at the same position. Nevertheless, the paper shows that the per-cluster power constraint exerts a substantial impact on the beamforming algorithm design. Fig. \ref{fig:nw_model} illustrates the considered downlink distributed antenna system with $C$ antenna clusters.

\begin{figure}[t]
\centering
\includegraphics[width=0.425\textwidth]{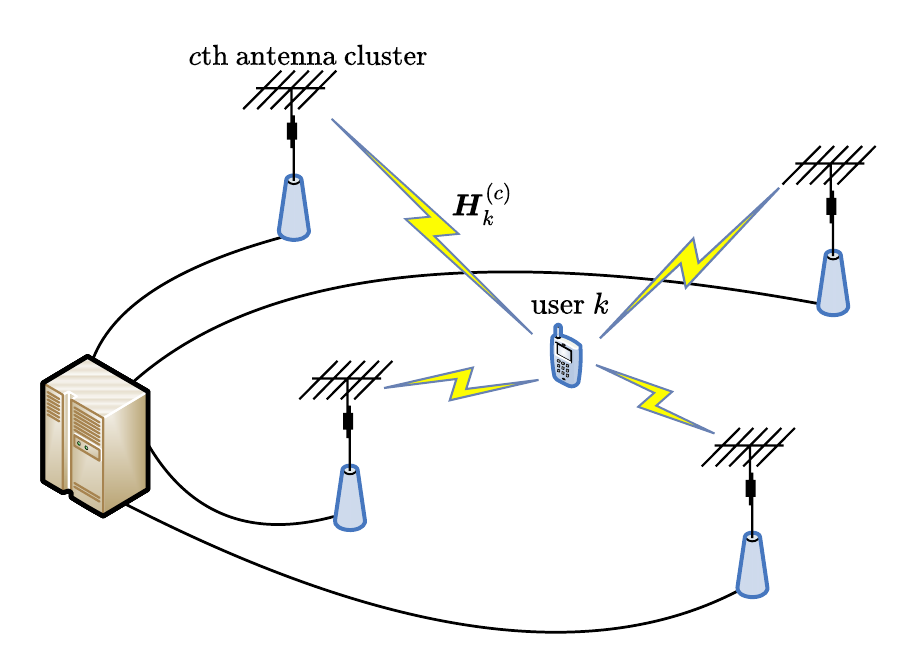}
\caption{Illustration of the downlink channel model with distributed antenna clusters. The BS has $N_t$ transmit antennas divided into $C$ clusters, each with $L$ antennas and power constraint $P_c$. User $k$ with $N_r$ receive antennas is shown as an example.}
\label{fig:nw_model}
\end{figure}

We address the beamforming problem for distributed antennas by means of manifold optimization. Manifold optimization has been recognized as a powerful tool for a broad range of constrained problems \cite{AbsilOpt}, \cite{boumal2023intromanifolds}. Some recent works have shown that the manifold optimization is well suited for the beamforming task in communication systems. In this area, the classic beamforming problem under the sum-power constraint is considered in \cite{sun2024RCG}, a joint passive and active beamforming optimization for a reconfigurable intelligent surface (RIS)-assisted network is considered in \cite{Shtaiwi2023RIS}, and the robust beamforming problem in the presence of channel estimation error is considered in \cite{wang2022robust3D}. We remark that all the above works formulate their beamforming problems as unconstrained problems on the spherical manifolds, so the classic techniques (such as the Riemannian conjugate gradient algorithm \cite{AbsilOpt}) can be readily applied.
Nevertheless, when the number of transmit antennas $N_t$ is large, optimizing directly over high-dimensional spherical manifolds can be computationally demanding, which motivates exploiting structure for dimensionality reduction. A key observation in \cite{zhao2023rethinking} is the low-dimensional subspace property, which reveals that stationary-point beamformers lie in a subspace spanned by the columns of users' channel matrices, enabling a significant dimensionality reduction from $N_t$ to $KN_r$, where $K$ is the number of users and $N_r$ is the number of receive antennas per user. However, this property has only been established for the sum-power constraint case. Motivated by this, we generalize the property to the per-cluster power constraint case. In contrast to the spherical formulations in the above works, the dimensionality reduction transforms the original spherical constraints into ellipsoidal ones, and it can be shown that the new feasible set forms a product ellipsoidal manifold. This change of set geometry, from spherical manifolds to ellipsoidal manifolds, is illustrated in Fig. \ref{fig:manifold_sets}. We then develop a tailored Riemannian conjugate gradient algorithm to solve the reformulated manifold optimization problem.



Related work on distributed antenna optimization under per-cluster power constraints has focused on various aspects. In practical systems such as coordinated multipoint (CoMP) transmission \cite{zhao2013CoMP}, BS antennas are grouped into several clusters, each subject to its own power budget due to hardware limitations or regulatory requirements. For spectral efficiency analysis, downlink performance under a stochastic model is studied in \cite{lin2014distributed}, revealing that selective transmission strategies outperform blanket transmission schemes. For beamforming optimization, a fractional-programming-based and solver-free accelerated algorithm for integrated sensing and communication (ISAC) beamforming under arbitrary per-cluster power constraints is proposed in \cite{Kim2025PGPCs}.
Other works have focused on practical system modeling and convex relaxation: a minimum-mean-square-error-based precoding approach under per-antenna-group constraints is formulated in \cite{mazrouei2013MMSEVP}, and a penalty-based convex framework for per-antenna and per-cluster power restrictions is introduced in \cite{WANG2013penalty}.

Most existing weighted sum-rate (WSR) maximization methods are developed for sum-power constraints, where the transmit power across all antennas is limited by a single budget. Early non-iterative linear precoders for the MIMO broadcast channel include maximum ratio transmission (MRT) \cite{lo1999MRT}, zero-forcing (ZF) \cite{gao2011zf}, and eigen zero-forcing \cite{sun2010ezf}. MRT maximizes received signal power but can incur severe inter-user interference, whereas ZF suppresses inter-user interference at the cost of noise amplification in ill-conditioned channels or at low SNR. Eigen zero-forcing exploits the singular value decomposition (SVD) of user channels to balance interference suppression and signal enhancement. While these methods are computationally efficient, they are suboptimal for WSR maximization, which motivates iterative algorithms such as weighted minimum mean square error (WMMSE) and fractional programming (FP) that jointly optimize transmit and receive variables for improved WSR.

The WMMSE algorithm \cite{christensen2008WMMSE}, \cite{shi2011WMMSE} reformulates the non-convex WSR problem into an equivalent mean square error (MSE) minimization problem solved via block coordinate descent (BCD) \cite{bertsekas1999nonlinear} with closed-form updates. However, its complexity scales quadratically with $N_t$, limiting scalability in massive MIMO. A recent work \cite{zhao2023rethinking} proposes a reduced WMMSE algorithm (referred to as RWMMSE in \cite{zhao2023rethinking}), exploiting a low-dimensional subspace property under sum-power constraints to solve a reduced MSE problem. FP methods \cite{shen2018fp1}, closely related to WMMSE \cite{shen2018fp2}, have also been applied to WSR maximization by reformulating the WSR problem via Lagrangian dual \cite{shen2019matrixFP} and quadratic transforms \cite{shen2024qt} into tractable BCD subproblems. However, the complexity of FP methods also scales quadratically with $N_t$. To tackle this issue, the FastFP method proposed in \cite{chen2025fastfp} employs a nonhomogeneous bound \cite{sun2017MM} to avoid large matrix inversions and the Nesterov's extrapolation \cite{nestrov2018lectures} to accelerate convergence.

\begin{figure}[t]
\centering
\includegraphics[width=0.46\textwidth]{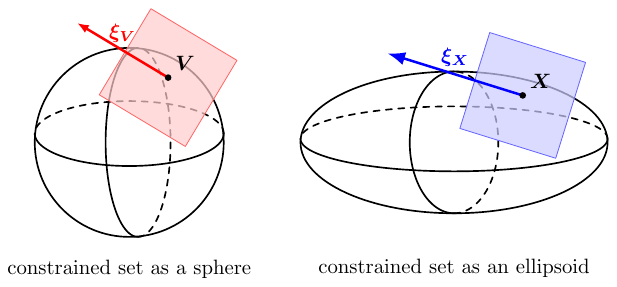}
\caption{Illustration of the feasible set in the original and low-dimensional reformulated problems. On the left-hand side, the feasible set forms a product of high-dimensional spherical manifolds in $\mathbb{C}^{N_t\times Kd_k}$. On the right-hand side, by leveraging the low-dimensional subspace property, the feasible set forms a product of high-dimensional ellipsoidal manifolds in $\mathbb{C}^{CKN_r\times Kd_k}$.}
\label{fig:manifold_sets}
\end{figure}

Motivated by the above discussion, the main contributions of this work are summarized as follows:
\begin{itemize}
    \item \textbf{Low-dimensional structure under per-cluster power constraints.} We generalize the low-dimensional subspace property established in \cite{zhao2023rethinking} for sum-power constraints to the more general per-cluster power constraints scenario. Specifically, we prove that, at any nontrivial stationary point, the optimal beamformers for each antenna cluster $c$ lie in the reduced subspace spanned by the columns of channels from cluster $c$. This extension reveals a structured low-dimensionality inherent to per-cluster power constraints, enabling significant dimensionality reduction while preserving optimality conditions, which is crucial for scalability in massive MIMO systems with clustered antenna architectures.
    \item \textbf{Reduced manifold reformulation via a product of ellipsoids.} Leveraging this property, we reformulate the nonconvex WSR maximization problem as an unconstrained optimization over a product of $C$ ellipsoid manifolds, a novel Riemannian submanifold tailored to per-cluster power constraints. We derive the complete Riemannian geometry for this manifold, including tangent spaces, metrics, orthogonal projections, retractions, and vector transports. Based on these, we develop an algorithm that incorporates Hestenes-Stiefel updates with Armijo backtracking line search for guaranteed convergence on manifolds, directly addressing the high-dimensional challenges unmet by prior manifold methods \cite{sun2024RCG} that are limited to sum-power constraints.
    \item \textbf{Proposed Riemannian conjugate gradient algorithm and scalability gains.} Through extensive numerical simulations in a single-cell massive MIMO downlink with large numbers of transmit antennas $N_t$, we demonstrate that our proposed algorithm achieves the same local optima as benchmarks (WMMSE \cite{shi2011WMMSE}, Riemannian conjugate gradient \cite{sun2024RCG}) under per-cluster power constraint scenarios, while reducing per-iteration complexity and elapsed time—particularly as $C$ increases. We further show that our algorithm consistently requires lower convergence time than Riemannian conjugate gradient \cite{sun2024RCG} across different system scales, including varying the number of users $K$ and the number of antennas per cluster $L$, demonstrating the scalability of the proposed low-dimensional reformulation.
\end{itemize}

The remainder of this paper is organized as follows. Section \ref{sec:S_M} presents the system model and formulates the WSR maximization problem under per-cluster power constraints. Section \ref{sec:convert_to_mani} establishes the low-dimensional subspace property and converts the original problem into an equivalent optimization over a product of ellipsoidal manifolds. Section \ref{sec:mani_algo} develops our proposed algorithm based on the Riemannian conjugate gradient algorithm. Section \ref{sec:nume_result} presents the numerical results. Finally, Section \ref{sec:conclusion} concludes the paper.

Here and throughout, boldface uppercase and lowercase letters denote matrices and vectors, respectively. For a vector $\ba$, $\ba^\hh$ is its conjugate transpose, and $\ba^\top$ is its transpose. For a matrix $\bA$, $\bA^\hh$ is its conjugate transpose, $\bA^\top$ is its transpose, $\|\bA\|_F$ is its Frobenius norm, $\Re(\bA)$ is its real part, $\operatorname{Col}(\bA)$ is its column space, $\operatorname{null}(\bA)$ is its null space, $\operatorname{rank}(\bA)$ is its rank, and $\mathrm{d}\bA$ is its matrix differential. For a square matrix $\bA$, $\tr(\bA)$ is its trace and $\det(\bA)$ is its determinant. The symbols $\bI$ and $\mathbf{0}$ denote the identity matrix and the zero matrix of appropriate dimensions. For a function $f(x)$, $\nabla f(x)$ is its Euclidean gradient. For a manifold $\mathcal{M}=\{\bX\in\mathbb{C}^{m\times n}:F(\bX)=0\}$, $F$ is the mapping: $\mathbb{C}^{m\times n}\to\mathbb{R}$. The differential of $F(\bX)$ is denoted by $\mathrm{D}F(\bX)$, and $\mathrm{D}F(\bX)[\bZ]$ represents the directional derivative of $F$ at $\bX$ along the direction $\bZ$.

\section{System Model}\label{sec:S_M}

Assume that the BS is associated with $K$ downlink user terminals. The BS has $N_t$ transmit antennas while each user terminal has $N_r$ receive antennas. Further, assume that the $N_t$ transmit antennas are divided into $C$ clusters; these clusters are placed at distributed positions. Note that each cluster has its own power constraint $P$. Each cluster comprises $L$ antennas, so we have $LC = N_t$. 

The channel from the $c$th antenna cluster to user $k$ is denoted by $\bH^{(c)}_k\in\mathbb C^{N_r\times L}$, for $c=1,\ldots,C$. These channel matrices are stacked horizontally with respect to each user $k$ as
\begin{equation}
    \bH_k = \begin{bmatrix}
        \bH^{(1)}_k & \bH^{(2)}_k & \ldots & \bH^{(C)}_k
    \end{bmatrix}\in\mathbb C^{N_r\times N_t}.
\end{equation}
Moreover, the channel matrices are stacked vertically with respect to each antenna cluster $c$ as
\begin{equation}
    \bH^{(c)} = \begin{bmatrix}
        \bH^{(c)}_1 \\ \vdots \\ \bH^{(c)}_K
    \end{bmatrix}\in\mathbb C^{KN_r\times L},
\end{equation}
which is assumed to be full row rank, i.e., $\operatorname{rank}(\bH^{(c)})=KN_r$.

Denote by $d_k\in\min\{N_t,N_r\}$ the number of downlink data streams assigned to user $k$. Let $\bV^{(c)}_{k}\in\mathbb C^{L\times d_k}$ be the transmit beamformer of the $c$th cluster of antennas for user $k$. The above beamforming matrices are vertically stacked for each user $k$ as
\begin{equation}
    \bV_k =
    \begin{bmatrix}
        \bV_{k}^{(1)} \\
        \vdots\\
        \bV_{k}^{(C)}
    \end{bmatrix}\in\mathbb C^{N_t\times d_k},
\end{equation}
and are horizontally stacked for each antenna cluster $c$ as
\begin{equation}
{\bV}^{(c)} = \left[\bV^{(c)}_1, \ldots, \bV^{(c)}_K\right] \in \mathbb{C}^{L \times Kd_k}.
\end{equation}
Moreover, denote by $\bs_k\in\mathbb C^{d_k}$ the modulated signal vector intended for user $k$, and $\bn_{k}\sim\mathcal {CN}(\mathbf{0}, \sigma^2_k\bI)$ the background noise at user $k$, where $\sigma^2_k$ is the noise power. Thus, the received signal of user $k$ can be computed as
\begin{align}
\by_k & =\underbrace{\bH_{k}\bV_{k}\bs_ {k}}_{\mathrm{desired\ signal}}+\underbrace{\bH_{k}\sum_{j=1, j\neq k}^{K}\bV_{j}\bs_{j}}_{\mathrm{multi-user\ interference}}+\bn_{k}.
\end{align}
An achievable rate $R_k$ for user $k$ is given by
\begin{equation}
R_k = \log\det\left(\bI+\bV_{k}^\hh\bH_{k}^\hh\bF_{k}^{-1}\bH_{k}\bV_{k}\right), \label{eq:rate}
\end{equation}
where
\begin{equation}
\bF_k=\sigma_k^2\bI+\sum_{j=1,j\neq k}^K\bH_k\bV_j\bV_j^\hh\bH_k^\hh.
\end{equation}
Each user $k$ is assigned a data rate weight $\omega>0$ in accordance with its priority. We seek the optimal beamformers $\{\bV_k\}$ to maximize the weighted sum rates
\begin{subequations}
\label{eq:original}
\begin{align}
\underset{{\bV}}{\text{maximize}} &\qquad \sum_{k=1}^K\omega_kR_k  \label{eq:wsr}\\
\text{subject to} & \qquad \|{\bV}^{(c)}\|_F^2\leq P_c,  \;\text{for}\;c=1,\ldots,C, \label{eq:pgpc}
\end{align}
\end{subequations}
where $P_c$ is the power constraint on the $c$th antenna cluster, and $\bV$ is the variable tuple $\begin{bmatrix}(\bV^{(1)})^\top , \ldots , (\bV^{(C)})^\top\end{bmatrix}^\top$.


\section{Conversion to Manifold Problems}\label{sec:convert_to_mani}

The goal of this section is to recast problem \eqref{eq:original} into a manifold optimization problem in order to get rid of the power constraint. We first show that the power constraint \eqref{eq:pgpc} must be tight for a stationary point solution, so it boils down to limiting the solution to a spherical manifold. Next, we propose further transforming the spherical manifold into an ellipsoidal manifold, thereby reducing the dimension and complexity of the beamforming problem.

\subsection{Spherical Manifold}

The key step is to show that the equality must hold in \eqref{eq:pgpc} when we have arrived at a stationary point solution of the problem, as specified in the following proposition.

\begin{proposition}
\label{prop:power tight}
    For any stationary solution $\bV^*$ to problem \eqref{eq:original}, the power constraint must be tight, i.e., 
    \begin{equation}
        \|({\bV}^{(c)})^*\|_F^2= P_c,\;\text{for}\;c=1,\ldots,C.
    \end{equation}
\end{proposition}
\begin{IEEEproof}
    See Appendix \ref{appendix:power_tight}.
\end{IEEEproof}

We remark that the special case of the above proposition when $C=1$ has been shown in the literature \cite{zhao2023rethinking, sun2024RCG}. Actually, this special case is quite easy to see: if $\|{\bV}^{(1)}\|_F^2< P_1$, then we can scale up all the ${\bV}^{(1)}_k$'s simultaneously to further increase $R_k$, so the current point cannot be a stationary point. However, the above argument no longer works when $C\ge2$. Thus, the result of Proposition \eqref{prop:power tight} is a nontrivial generalization.

We then rewrite problem \eqref{eq:original} as
\begin{subequations}
\label{eq:equality}
\begin{align}
\underset{{\bV}}{\text{maximize}}  &\quad \sum_{k=1}^K\omega_kR_k  \\
\text{subject to} &\quad \|{\bV}^{(c)}\|_F^2=P_c,  \;\text{for}\;c=1,\ldots,C.
\end{align}
\end{subequations}
The new constraint \eqref{eq:original} can be interpreted as a manifold condition. Specifically, each power constraint $P_c$ induces a spherical manifold:
\begin{equation}
    \mathcal{S}_c = \left\{{\bV}^{(c)} : \|{\bV}^{(c)}_k\|_F^2=P_c\right\}.
\end{equation}
Accordingly, \eqref{eq:original} boils down to a product manifold:
\begin{equation}
\mathcal{S} = \mathcal{S}_1 \times \mathcal{S}_2 \times \cdots \times \mathcal{S}_C.
\end{equation}
As such, the constrained problem \eqref{eq:equality} can be rewritten as an unconstrained problem with $\bV$ limited to $\mathcal S$:
\begin{equation}
\underset{{\bV}\in\mathcal{S}}{\text{minimize}} \quad -\sum_{k=1}^K\omega_kR_k.\label{eq:spheres}
\end{equation}
Note that we write this new problem in the minimization form by the convention of manifold optimization \cite{AbsilOpt}.

For the single-cluster case with $C=1$ in \cite{sun2024RCG}, $\mathcal{S}$ reduces to a spherical manifold---which has been extensively considered in the literature \cite{zhao2023rethinking}. For the general product manifold $\mathcal S$ with $C\ge2$, we can extend the existing algorithm for it because it is an embedded submanifold of the ambient Euclidean space $\mathbb C^{N_t\times Kd_k}$ \cite{AbsilOpt}.

\subsection{Ellipsoidal Manifold}

From a manifold optimization perspective, the new problem \eqref{eq:spheres} is easier to tackle since the constraint has been removed. We now show that the problem can be further simplified by converting the spherical manifold $\mathcal M$ to an ellipsoidal manifold of a lower dimension. 

Introduce a new variable $\bX^{(c)}_k\in\mathbb C^{N_r\times d_k}$ that corresponds to each $\bV^{(c)}_k\in\mathbb C^{L\times d_k}$. Similarly, we define
\begin{equation}
    \bX_k =
    \begin{bmatrix}
        \bX_{k}^{(1)} \\
        \vdots\\
        \bX_{k}^{(C)}
    \end{bmatrix}\in\mathbb C^{CN_r\times d_k}
\end{equation}
and 
\begin{equation}
{\bX}^{(c)} = \left[\bX^{(c)}_1, \ldots, \bX^{(c)}_K\right] \in \mathbb{C}^{KN_r \times Kd_k},
\end{equation}
and the tuple $({\bX}^{(1)},\ldots,{\bX}^{(C)})$ is denoted by $\bX$. By substituting
\begin{equation}
    \bV^{(c)}_k = (\bH^{(c)})^\hh\bX^{(c)}_k \label{eq:Vkc_transform}
\end{equation}
into problem \eqref{eq:spheres}, we obtain a manifold problem of $\bX$:
\begin{equation}
\underset{{\bX}\in\mathcal{M}}{\text{minimize}} \quad -\sum_{k=1}^K\omega_k\hat{R}_k, \label{eq:reform ellipse}
\end{equation}
where the data rate is now computed as
\begin{multline}
    \hat{R}_k = \log\det\Biggl(\bI+ \bX_{k}^\hh{\bG}_{k}^\hh\\
    \biggl(\sigma_k^2\bI +\sum_{j=1,j\neq k}^K {\bG}_{k}\bX_j\bX_j^\hh{\bG}_{k}^\hh\biggr)^{-1}{\bG}_{k}\bX_{k}\Biggr) \label{eq:rate_reform}
\end{multline}
with
\begin{equation}
    \bG_k = \left[\bH_{k}^{(1)} (\bH^{(1)})^\hh,\ldots, \bH_{k}^{(C)} (\bH^{(C)})^\hh \right],
\end{equation}
and the product manifold now becomes
\begin{equation}
\mathcal{M} = \mathcal{M}_1 \times \mathcal{M}_2 \times \cdots \times \mathcal{M}_C,
\label{eq:Manifold}
\end{equation}
where each $\mathcal M_c$ is an ellipsoidal manifold:
\begin{equation}
\mathcal{M}_c = \left\{{\bX}^{(c)}:\operatorname{tr}\left(({\bX}^{(c)})^\hh \bQ^{(c)}{\bX}^{(c)}\right)=P_c\right\},
\label{eq:power_reform}
\end{equation}
where 
\begin{equation}
    \bQ^{(c)} = \bH^{(c)}(\bH^{(c)})^\hh.
\end{equation}

\begin{proposition}
\label{prop:ellipsoid}
The set $\mathcal{M}_c$ defined in \eqref{eq:power_reform} is a smooth embedded submanifold of the ambient Euclidean space $\mathbb{C}^{KN_r\times Kd_k}$. 
\end{proposition}
\begin{IEEEproof}
See Appendix~\ref{appendix:ellipsoid}.
\end{IEEEproof}

The following proposition is the building block of this section:

\begin{proposition}
\label{prop:reform_equiv}
Problem \eqref{eq:reform ellipse} is equivalent to problem \eqref{eq:equality} in the sense that $\bX$ is a stationary point solution of problem \eqref{eq:reform ellipse} if and only if the corresponding $\bV$ is a stationary point solution of problem \eqref{eq:equality}.
\end{proposition}
\begin{IEEEproof}
  See Appendix \ref{appendix:reform_equiv}.
\end{IEEEproof}

We remark that the above theorem is foreshadowed by the result of \cite{zhao2023rethinking}, which aims to reduce the variable dimension for the WMMSE algorithm when $C=1$. Most importantly, because $N_r\ll L$, the dimension and complexity of the manifold optimization in \eqref{eq:reform ellipse} are much lower than those of the spherical manifold case in \eqref{eq:spheres}.

\section{Manifold Optimization Algorithm}\label{sec:mani_algo}

We propose using the conjugate gradient method \cite{AbsilOpt} to address the manifold optimization problem in \eqref{eq:reform ellipse}. This section consists of two parts: first, we sketch the overall procedure of the Riemannian conjugate gradient method on a Riemannian manifold, and review some terminologies it involves; second, we specify the method for our problem case.

\subsection{Preliminaries}
\label{sec:preliminaries}
Consider a general optimization problem on a Riemannian manifold $\mathcal{M}$:
\begin{equation}
\underset{x\in\mathcal M}{\text{minimize}} \quad f(x).\label{prob:general}
\end{equation}
To ease illustration, let us further assume that $\mathcal{M}$ is an embedded submanifold of the complex Euclidean space $\mathbb{C}^n$ defined by a smooth function $F: \mathbb{C}^n \to \mathbb{R}$ as
\begin{equation}
\mathcal{M} = \left\{x : F(x)=0 \right\}.
\end{equation}
The Riemannian conjugate gradient method updates $x$ iteratively on the manifold. Each iteration consists of two parts: first, update $x$ by the Riemannian gradient, so the new $x$ departs from $\mathcal M$; second, retract the new $x$ back onto $\mathcal M$. We detail it in the following.

\newcounter{MYtempeqncnt2}
\setcounter{MYtempeqncnt2}{\value{equation}}
\begin{figure*}[t]
\setcounter{equation}{33}
\begin{equation}
    \beta_{\ell} = \frac{\langle \mathrm{grad}\, f(x_{\ell-1}),\ \mathrm{grad}\, f(x_{\ell-1}) - \mathcal{T}(\mathrm{grad}\, f(x_{\ell-2}), \alpha_{\ell-2}) \rangle_{x_{\ell-1}}}{\langle \mathcal{T}(\eta_{\ell-1}, \alpha_{\ell-2}),\ \mathrm{grad}\, f(x_{\ell-1}) - \mathcal{T}(\mathrm{grad}\, f(x_{\ell-2}), \alpha_{\ell-2}) \rangle_{x_{\ell-1}}} \label{eq:beta_l}
\end{equation}
\hrulefill
\end{figure*}
\setcounter{equation}{\value{MYtempeqncnt2}}
\subsubsection{Remannian Gradient}

We start by introducing some concepts of the Riemannian geometry. The \emph{tangent space} to $\mathcal{M}$ at a point $x \in \mathcal{M}$ is the set of all tangent vectors of $F(x)$ at the point $x$:
\begin{equation}
    T_x\mathcal{M} = \left\{\xi \in \mathbb{C}^n : \mathrm{D}F(x)[\xi] = 0\right\}.
\end{equation}
Now, consider any two vectors $\xi$ and $\eta$ in the tangent space, namely \emph{tangent vectors}. The \emph{Riemannian metric} between them is computed as the inner product:
\begin{equation}
\langle \xi, \eta \rangle_x = \mathrm{Re}\left\{\xi^\hh \eta\right\}.
\end{equation}
The orthogonal complement of the tangent space $T_x\mathcal{M}$ is called the \emph{normal space}, spanned by the constraint gradient:
\begin{equation}
N_x\mathcal{M} = \left\{\nu \in \mathbb{C}^n : \nu = \lambda \nabla F(x), \lambda \in \mathbb{R}\right\}.
\end{equation}
Since $T_x\mathcal{M}$ and $N_x\mathcal{M}$ form a direct sum of $\mathcal C^n$, an arbitrary vector $z \in \mathbb{C}^n$ can be uniquely partitioned into a pair of vectors belonging to the two subspaces respectively, i.e., \emph{orthogonal projection}. Specifically, the projection of $z$ onto the tangent space $T_x\mathcal{M}$ is given by
\begin{equation}
\mathcal{P}_x(z) = z - \frac{\langle \nabla F(x), z \rangle_x}{\langle \nabla F(x), \nabla F(x) \rangle_x} \nabla F(x).
\end{equation}
The \emph{Riemannian gradient} is obtained by projecting the Euclidean gradient $\nabla f(x)$ in $\mathbb C^n$ onto the tangent space:
\begin{equation}
\mathrm{grad}\, f(x) = \mathcal{P}_x(\nabla f(x)). \label{eq:Riemannian_gradient}
\end{equation}
But the Riemannian conjugate gradient method modifies the gradient further. Denote this modified Riemannian gradient of the $\ell$th iteration by $\eta_\ell$. Thus, for the $\ell$th iteration, we update $x$ as
\begin{equation}
 \hat x_{\ell} =  x_{\ell-1} + \alpha_{\ell}\eta_{\ell}, \label{eq:update_xhat}
\end{equation}
where $\alpha_{\ell}>0$ is the update stepsize which can be determined by the backtracking line search with the Armijo condition \cite{Armijo1966}.

\begin{figure}[t]
    \centering
    \subfloat[Illustration of the vector transport $\mathcal{T}(\eta_{\ell-1},\alpha_{\ell-1})$ and the search direction $\eta_\ell$. The vector transport moves $\eta_{\ell-1}$ from $T_{x_{\ell-2}}\mathcal{M}$ to $T_{x_{\ell-1}}\mathcal{M}$, and the search direction $\eta_\ell$ is a combination of the negative Riemannian gradient and the transported vector.\label{fig:transport}]{%
        \includegraphics[width=0.9\columnwidth]{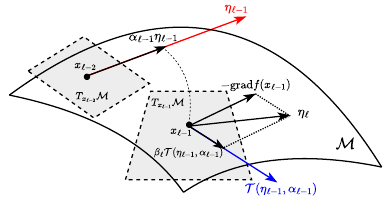}
    }\\
    \subfloat[Illustration of the update of $x_{\ell}$ in the Riemannian conjugate gradient method. The new point $\hat x_{\ell}$ is obtained by moving from $x_{\ell-1}$ along the search direction $\eta_\ell$ with stepsize $\alpha_\ell$, and then retracted back to the manifold to get $x_\ell$.\label{fig:updatex}]{%
        \includegraphics[width=0.9\columnwidth]{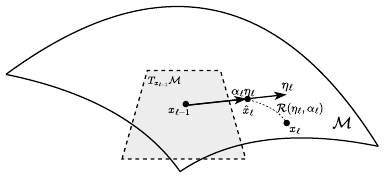}
    }
    \caption{Illustrations of the Riemannian conjugate gradient method.}
    \label{fig:combined}
\end{figure}

\newcounter{MYtempeqncnt}
\setcounter{MYtempeqncnt}{\value{equation}}
\begin{figure*}[t]
\setcounter{equation}{38}
\vspace*{0pt}
\begin{equation}
\begin{aligned}
\mathcal{P}_x(\xi) &=
\begin{bmatrix}
\mathcal{P}_{x^{(1)}}(\xi^{(1)}) \\
\vdots \\
\mathcal{P}_{x^{(n)}}(\xi^{(n)})
\end{bmatrix},
\quad
\mathcal{R}(\eta, \alpha_{\ell}) =
\begin{bmatrix}
\mathcal{R}(\eta^{(1)}, \alpha_{\ell}) \\
\vdots \\
\mathcal{R}(\eta^{(n)}, \alpha_{\ell})
\end{bmatrix},
\quad
\mathcal{T}(\eta, \alpha_{\ell}) =
\begin{bmatrix}
\mathcal{T}(\eta^{(1)}, \alpha_{\ell}) \\
\vdots \\
\mathcal{T}(\eta^{(n)}, \alpha_{\ell})
\end{bmatrix}
\end{aligned}
\label{eq:product_opes}
\end{equation}
\hrulefill
\vspace*{0pt}
\setcounter{equation}{45}
\begin{equation}
{\bX}_{\ell}^{(c)}=\mathcal{R}(\boldsymbol{\eta}_{\ell}^{(c)}, \alpha_{\ell}) = \sqrt{{P_c}\bigg/{\operatorname{tr}\left(\bQ^{(c)}\left({\bX}_{\ell-1}^{(c)}+\alpha_{\ell}\boldsymbol{\eta}_{\ell}^{(c)}\right)\left({\bX}_{\ell-1}^{(c)}+\alpha_{\ell}\boldsymbol{\eta}_{\ell}^{(c)}\right)^\hh\right)}}\left({\bX}_{\ell-1}^{(c)}+\alpha_{\ell}\boldsymbol{\eta}_{\ell}^{(c)}\right)\label{eq:retra}
\end{equation}
\vspace*{0pt}
\hrulefill
\setcounter{equation}{47}
\begin{equation}
\lambda' = \frac{\operatorname{tr}\left(\bQ^{(c)}\left(\boldsymbol{\xi}^{(c)}({\bX}^{(c)}+\boldsymbol{\eta}^{(c)})^\hh+({\bX}^{(c)}+\boldsymbol{\eta}^{(c)})(\boldsymbol{\xi}^{(c)})^\hh\right)\right)}{2\operatorname{tr}\left(\bQ^{(c)}\bQ^{(c)}({\bX}^{(c)}+\boldsymbol{\eta}^{(c)})({\bX}^{(c)}+\boldsymbol{\eta}^{(c)})^\hh\right)}\label{eq:lambda_prime}
\end{equation}
\vspace*{0pt}
\hrulefill
\setcounter{equation}{48}
\begin{equation}
\begin{aligned}
\mathcal{P}_{{\bX}}(\bZ) &=
\begin{bmatrix}
\mathcal{P}_{{\bX}^{(1)}}(\bZ^{(1)}) \\
\vdots \\
\mathcal{P}_{{\bX}^{(C)}}(\bZ^{(C)})
\end{bmatrix},
\quad
\mathcal{R}(\boldsymbol{\eta}, \alpha_{\ell})\! &=
\begin{bmatrix}
\mathcal{R}(\boldsymbol{\eta}^{(1)}, \alpha_{\ell}) \\
\vdots \\
\mathcal{R}(\boldsymbol{\eta}^{(C)}, \alpha_{\ell})
\end{bmatrix},
\quad
\mathcal{T}(\boldsymbol{\eta}, \alpha_{\ell}) &=
\begin{bmatrix}
\mathcal{T}(\boldsymbol{\eta}^{(1)}, \alpha_{\ell}) \\
\vdots \\
\mathcal{T}(\boldsymbol{\eta}^{(C)}, \alpha_{\ell})
\end{bmatrix}
\end{aligned}
\label{eq:product_operations}
\end{equation}
\hrulefill
\vspace*{0pt}
\end{figure*}
\setcounter{equation}{\value{MYtempeqncnt}}

\subsubsection{Retraction} However, $ \hat x_{\ell}$ is not guaranteed to be on the manifold. To make the iterative update feasible, we need to retract it back onto $\mathcal M$:
\begin{equation}
x_{\ell}=\mathcal{R}(\eta_{\ell},\alpha_\ell): = \arg\min_{y \in \mathcal{M}} \|y - \hat x_{\ell}\|_F. \label{eq:retraction}
\end{equation}

Note that $\alpha_\ell$ is an argument of the retraction $\mathcal R(\cdot)$ because $\hat x_\ell$ depends on it.
It remains to specify how the Riemannian conjugate gradient method changes the original gradient $\mathrm{grad}\, f(x)$ to $\eta_\ell$; this change involves the retraction operation and a new notion called vector transport. The idea of \emph{vector transport} is to move $\eta_{\ell-1}$ of the previous iteration from $T_{x_{\ell-2}}\mathcal{M}$ to $T_{x_{\ell-1}}\mathcal{M}$. Following \cite{AbsilOpt}, it is computed as
\begin{equation}
\mathcal{T}(\eta_{\ell-1},\alpha_{\ell-1}) =  \mathcal{P}_{x_{\ell-1}}(\eta_{\ell-1})=\mathcal{P}_{\mathcal{R}(\eta_{\ell-1},\alpha_{\ell-1})}(\eta_{\ell-1}). \label{eq:x_transport}
\end{equation}
We finally obtain $\eta_l$ as
\begin{equation}
    \eta_{\ell} = -\mathrm{grad}\, f(x_{\ell-1}) + \beta_{\ell} \mathcal{T}(\eta_{\ell-1},\alpha_{\ell-1}), \label{eq:eta_l}
\end{equation}
where the conjugate coefficient $\beta_{\ell}$ is determined by the Hestenes-Stiefel formula \cite{Salleh2016HSbeta} given in \eqref{eq:beta_l} at the top of the page. \addtocounter{equation}{1}
As mentioned in \eqref{eq:update_xhat}, the stepsize $\alpha_{\ell}$ is determined via backtracking line search with the Armijo condition:
\begin{multline}
 f(\mathcal{R}(\eta_{\ell}, \sigma^m \alpha^0 )) \leq f(x_{\ell-1}) \\+ p\sigma^m \alpha^0 \langle \mathrm{grad}\, f(x_{\ell-1}), \eta_{\ell} \rangle_{x_{\ell-1}}, \label{eq:armijo_l}
\end{multline}
where the initial stepsize $\alpha^0>0$ is iteratively reduced by a factor $\sigma\in(0,1)$ until the condition is satisfied, and $p\in(0,1)$ is a constant that controls the sufficient decrease. The stepsize is then determined as $\alpha_{\ell} = \sigma^m \alpha^0$. Fig. \ref{fig:combined} illustrates the above operations in the Riemannian conjugate gradient method. 


\subsubsection{Product Manifold}

The above algorithm can be further extended for the product manifold case. Assume now that the manifold becomes:
\begin{equation}
\mathcal{M} = \mathcal{M}_1 \times \ldots \times \mathcal{M}_n,
\end{equation}
where each $\mathcal M_i$ is a Riemannian manifold. The above terminologies can be readily carried over to the product manifold case. For a point $x = (x^{(1)}, \ldots, x^{(n)}) \in \mathcal{M}$, the tangent space is now given by
\begin{equation}
T_x\mathcal{M} = T_{x^{(1)}}\mathcal{M}_1 \oplus T_{x^{(2)}}\mathcal{M}_2 \oplus \cdots \oplus T_{x^{(n)}}\mathcal{M}_n.
\end{equation}
The Riemannian metric then amounts to the sum of the Riemannian  metrics across the submanifolds:
\begin{equation}
\langle \xi, \eta \rangle_x = \sum_{i=1}^n \langle \xi^{(i)}, \eta^{(i)} \rangle_{x^{(i)}},
\end{equation}
where $\xi = (\xi^{(1)}, \ldots, \xi^{(n)})$ and $\eta = (\eta^{(1)}, \ldots, \eta^{(n)})$ are tangent vectors at $x$. Similarly, the orthogonal projection, retraction, and transport on the product manifold are computed component-wise as in \eqref{eq:product_opes}.
\addtocounter{equation}{1}

\subsection{Reduced Riemannian Conjugate Gradient Method}

We now apply the Riemannian conjugate gradient method to problem \eqref{eq:reform ellipse}. Following the framework in Subsection~\ref{sec:preliminaries}, we specify the geometric elements and gradient computation for the ellipsoidal product manifold $\mathcal{M} = \mathcal{M}_1 \times \cdots \times \mathcal{M}_C$.

\subsubsection{Geometric Elements}

For each component manifold $\mathcal{M}_c$, the tangent space at ${\bX}^{(c)} \in \mathcal{M}_c$ is
\begin{equation}
T_{{\bX}^{(c)}}\mathcal{M}_c =\left\{\boldsymbol{\xi}^{(c)}:\Re\left\{\operatorname{tr}\left(\bQ^{(c)}\left(\boldsymbol{\xi}^{(c)}({\bX}^{(c)})^\hh\right)\right)\right\}=0\right\},
\end{equation}
and the normal space is
\begin{equation}
N_{{\bX}^{(c)}}\mathcal{M}_c= \left\{\boldsymbol{\eta}^{(c)}: \boldsymbol{\eta}^{(c)} = \lambda \bQ^{(c)}{\bX}^{(c)}, \lambda \in \mathbb{R}\right\}.
\end{equation}
For any two tangent vectors $\boldsymbol{\xi}^{(c)}, \boldsymbol{\eta}^{(c)} \in T_{{\bX}^{(c)}}\mathcal{M}_c$, the Riemannian metric is 
\begin{equation}
\left\langle\boldsymbol{\xi}^{(c)}, \boldsymbol{\eta}^{(c)}\right\rangle _{{\bX}^{(c)}} = \Re\left\{\operatorname{tr}\left((\boldsymbol{\xi}^{(c)})^\hh\boldsymbol{\eta}^{(c)}\right)\right\}. 
\end{equation}
The orthogonal projection of an arbitrary vector $\bZ^{(c)} \in \mathbb{C}^{N_r \times Kd_k}$ onto the tangent space $T_{{\bX}^{(c)}}\mathcal{M}_c$ is given by
\begin{equation}
\mathcal{P}_{{\bX}^{(c)}}(\bZ^{(c)}) = \bZ^{(c)} - \lambda\bQ^{(c)}{\bX}^{(c)},
\end{equation}
where $\lambda$ is chosen such that $\mathcal{P}_{{\bX}^{(c)}}(\bZ^{(c)}) \in T_{{\bX}^{(c)}}\mathcal{M}_c$, i.e.,
\begin{equation}
\lambda = \frac{\operatorname{tr}\left(\bQ^{(c)}\left(\bZ^{(c)}({\bX}^{(c)})^\hh+{\bX}^{(c)}(\bZ^{(c)})^\hh\right)\right)}{2\operatorname{tr}\left(\bQ^{(c)}\bQ^{(c)}{\bX}^{(c)}({\bX}^{(c)})^\hh\right)}.
\end{equation}
Given a search direction $\boldsymbol{\eta}_{\ell}^{(c)} \in T_{{\bX}^{(c)}}\mathcal{M}_c$ and stepsize $\alpha_{\ell} > 0$ in ${\ell}$th iteration, we update ${\bX}_{\ell}^{(c)}$ as
\begin{equation}
 \hat{\bX}_{\ell}^{(c)} =  {\bX}_{\ell-1}^{(c)} + \alpha_{\ell}\boldsymbol{\eta}_{\ell}^{(c)}.
\end{equation}
We then retract $\hat{\bX}_{\ell}^{(c)}$ back onto $\mathcal{M}_c$ using the retraction defined in \eqref{eq:retra} at the top of the page.\addtocounter{equation}{1}

To transport $\boldsymbol{\eta}_{\ell-1}^{(c)}$ from $T_{{\bX}_{\ell-2}^{(c)}}\mathcal{M}_c$ to $T_{{\bX}_{\ell-1}^{(c)}}\mathcal{M}_c$, we compute the transport as
\begin{equation}
\begin{aligned}
\mathcal{T}(\boldsymbol{\eta}_{\ell-1}^{(c)}, \alpha_{\ell-1}) &= \mathcal{P}_{\mathcal{R}(\boldsymbol{\eta}_{\ell-1}^{(c)}, \alpha_{\ell-1})}\left(\boldsymbol{\eta}_{\ell-1}^{(c)}\right)\\
&= \boldsymbol{\eta}_{\ell-1}^{(c)} - \lambda'\bQ^{(c)}\left({\bX}_{\ell-1}^{(c)}+\boldsymbol{\eta}_{\ell-1}^{(c)} \right),
\end{aligned}
\end{equation}
where $\lambda'$ is given in \eqref{eq:lambda_prime} at the top of the page. For the product manifold $\mathcal{M}$, the projection, retraction, and transport are computed component-wise in \eqref{eq:product_operations}. \addtocounter{equation}{2}

\newcounter{MYtempeqncnt1}
\setcounter{MYtempeqncnt1}{\value{equation}}
\begin{figure*}[!t]
\setcounter{equation}{49}
\begin{align}
\nabla f(\bX_k) &= -\omega_k \nabla_{\bX_k} \hat{R}_k - \sum_{i=1, i\neq k}^K \omega_i \nabla_{\bX_k} \hat{R}_i\notag\\
&= -\omega_k \bG_k^\hh \left(\sigma_k^2\bI +\sum_{j=1}^K \bG_{k}\bX_j\bX_j^\hh\bG_{k}^\hh\right)^{-1}\bG_{k}\bX_{k}\notag\\
&\qquad - \sum_{i=1,i\neq k}^K \omega_i \bG_i^\hh \left[\left(\sigma_i^2\bI +\sum_{j=1}^K \bG_{i}\bX_j\bX_j^\hh\bG_{i}^\hh\right)^{-1}- \left(\sigma_i^2\bI +\sum_{j=1, j\neq i}^K \bG_{i}\bX_j\bX_j^\hh\bG_{i}^\hh\right)^{-1}\right] \bG_{i}\bX_{k}
\label{eq:euclidean_grad}
\end{align}
\hrulefill
\vspace*{1pt}
\setcounter{equation}{53}
\begin{equation}
\beta_{\ell} = \frac{\left\langle \mathrm{grad}\, f({\bX}_{\ell-1}),\ \mathrm{grad}\, f({\bX}_{\ell-1}) - \mathcal{T}(\mathrm{grad}\, f({\bX}_{\ell-2}), \alpha_{\ell-2}) \right\rangle_{{\bX}_{\ell-1}}}
{\left\langle\mathcal{T}(\boldsymbol{\eta}_{\ell-1}, \alpha_{\ell-2}),\ \mathrm{grad}\, f({\bX}_{\ell-1}) - \mathcal{T}(\mathrm{grad}\, f({\bX}_{\ell-2}), \alpha_{\ell-2}) \right\rangle_{{\bX}_{\ell-1}}}\label{eq:betaHS}
\end{equation}
\hrulefill
\end{figure*}
\setcounter{equation}{\value{MYtempeqncnt1}}

\subsubsection{Riemannian Gradient}

Define the objective function in \eqref{eq:reform ellipse} as $f({\bX}) := -\sum_{k=1}^K \omega_k \hat{R}_k$. Recall the tuple ${\bX} = \left[\bX_1, \ldots, \bX_K\right]$, we first compute the Euclidean gradient of $f({\bX})$ with respect to each variable $\bX_k$. The Euclidean gradient is $\nabla f({\bX}) = \left[\nabla f(\bX_1), \ldots, \nabla f(\bX_K)\right]$, where $\nabla f(\bX_k)$ is given in \eqref{eq:euclidean_grad} at the top of next page. \addtocounter{equation}{1} Note that $\nabla f(\bX_k)$ is vertically stacked as
\begin{equation}
\nabla f(\bX_k) = \left[\nabla f(\bX_k^{(1)})^\top, \nabla f(\bX_k^{(2)})^\top, \ldots, \nabla f(\bX_k^{(C)})^\top\right]^\top.
\end{equation}
The Riemannian gradient is then obtained by projection:
\begin{equation}
\mathrm{grad}\, f({\bX}) = \mathcal{P}_{{\bX}}(\nabla f({\bX})).\label{eq:rie_grad}
\end{equation}

\subsubsection{Conjugate Search Direction and Line Search}

We construct the search direction $\boldsymbol{\eta}_{\ell}$ as
\begin{equation}
\boldsymbol{\eta}_{\ell} = -\mathrm{grad}\, f({\bX}_{\ell-1}) + \beta_{\ell} \,\mathcal{T}\!\left(\boldsymbol{\eta}_{\ell-1}, \alpha_{\ell-1}\right), \label{eq:search_direction}
\end{equation}
where $\beta_{\ell}$ is the Hestenes-Stiefel parameter \cite{Salleh2016HSbeta} given in \eqref{eq:betaHS} at the top of next page. To ensure descent, we set
\addtocounter{equation}{1}
\begin{equation}
\beta_{\ell}=
\begin{cases}
\max\{\beta_\ell^{HS}, 0\}, & \text{if } \left\langle \mathrm{grad}\, f({\bX}_{\ell-1}), \boldsymbol{\eta}_{\ell}\right\rangle_{{\bX}_{\ell-1}} > 0,\\
0, & \text{otherwise}.
\end{cases}
\label{eq:beta_modified}
\end{equation}
The stepsize $\alpha_{\ell}$ is determined by Armijo backtracking: find the smallest non-negative integer $m$ such that
\begin{multline}
f(\mathcal{R}(\boldsymbol{\eta}_{\ell}, \sigma^m \alpha^0)) \leq f({\bX}_{\ell-1})\\ + \sigma^m p\ \alpha^0 \left\langle \mathrm{grad}\, f({\bX}_{\ell-1}), \boldsymbol{\eta}_{\ell} \right\rangle_{{\bX}_{\ell-1}}, \label{eq:armijo}
\end{multline}
and set $\alpha_{\ell} = \sigma^m \alpha^0$.
The complete algorithm is summarized in Algorithm~\ref{alg:RRCG}.
\begin{algorithm}[t]
  \caption{Proposed Distributed Antenna Beamforming}
  \label{alg:RRCG}
  \begin{algorithmic}[1]
      \STATE Initialize ${\bX}_0\in{\mathcal{M}}$, set $\ell=1$, and choose parameters $\sigma, p \in (0,1)$, $\alpha^0>0$.
      \STATE Set the initial search direction $\boldsymbol{\eta}_0 = -\mathrm{grad}\, f({\bX}_0)$.
      \REPEAT
        \STATE Compute the Riemannian gradient $\mathrm{grad}\, f({\bX}_{\ell-1})$ using \eqref{eq:rie_grad}.
        \IF {$\ell=1$}
          \STATE Set the search direction $\boldsymbol{\eta}_1 = -\mathrm{grad}\, f({\bX}_{0})$.
        \ELSE
          \STATE Compute the vector transport $\mathcal{T}(\boldsymbol{\eta}_{\ell-1}, \alpha_{\ell-1})$.
          \STATE Compute $\beta_{\ell}$ using \eqref{eq:beta_modified}.
          \STATE Update the search direction $\boldsymbol{\eta}_\ell$ using \eqref{eq:search_direction}.
        \ENDIF 
        \STATE Find stepsize $\alpha_\ell$ via backtracking line search with the Armijo condition \eqref{eq:armijo}.
        \STATE Retract onto the manifold using \eqref{eq:product_operations} and obtain ${\bX}_\ell$.
        \STATE Increment $\ell \leftarrow \ell + 1$.
      \UNTIL{the value of $f({\bX}_\ell)$ converges}
  \end{algorithmic}
\end{algorithm}

\subsubsection{Convergence Analysis}

We now establish the convergence of the proposed algorithm to a stationary point.

\begin{proposition}
\label{prop:convergence}
The sequence $\{{\bX}_\ell\}_{\ell=0}^\infty$ generated by Algorithm~\ref{alg:RRCG} converges to a stationary point of $f$ on $\mathcal{M}$, i.e.,
\begin{equation}
\lim_{\ell \to \infty} \left\|\mathrm{grad}\, f({\bX}_\ell)\right\|_{F} = 0.
\end{equation}
\end{proposition}

\begin{IEEEproof}
We first show that the search direction $\boldsymbol{\eta}_\ell$ is a descent direction. By construction in \eqref{eq:search_direction}, when $\ell = 1$, we have $\boldsymbol{\eta}_1 = -\mathrm{grad}\, f({\bX}_0)$, so
\begin{equation}
\left\langle \mathrm{grad}\, f({\bX}_0), \boldsymbol{\eta}_1 \right\rangle_{{\bX}_0} = -\left\|\mathrm{grad}\, f({\bX}_0)\right\|^2_{{\bX}_0} < 0.
\end{equation}
For $\ell \geq 2$, the modification of $\beta_\ell$ in \eqref{eq:beta_modified} ensures that $\left\langle \mathrm{grad}\, f({\bX}_{\ell-1}), \boldsymbol{\eta}_\ell \right\rangle_{{\bX}_{\ell-1}} \leq 0$. Thus, $\boldsymbol{\eta}_\ell$ is a descent direction for all $\ell$.

Since the search direction is a descent direction and the Armijo line search \eqref{eq:armijo} is employed, by \cite[Theorem~4.3.1]{AbsilOpt}, every accumulation point of the sequence $\{{\bX}_\ell\}$ is a stationary point of $f$ on $\mathcal{M}$.

From \cite[Theorem 1.4.8]{Conway2014PointSet}, a subset of a finite dimensional Euclidean space is compact if and only if it is closed and bounded. Since each component manifold $\mathcal{M}_c$ is defined by the constraint $\operatorname{tr}(\bQ^{(c)}{\bX}^{(c)}({\bX}^{(c)})^\hh) = P_c$, it is closed and bounded according to \cite[Definition 1.1.6]{Conway2014PointSet} and \cite[Definition 1.2.15]{Conway2014PointSet}. Therefore, the component manifold $\mathcal{M}_c$ and the product manifold $\mathcal{M}$ is compact. By \cite[Corollary~4.3.2]{AbsilOpt}, the compactness of $\mathcal{M}$ implies that $\lim_{\ell \to \infty} \|\mathrm{grad}\, f({\bX}_\ell)\|_{F} = 0$.
\end{IEEEproof}

\section{Numerical Results}
\label{sec:nume_result}
\subsection{Simulation Setup}
In this section, we present simulation results to evaluate the performance of proposed algorithm.
In the single-cell downlink system with a hexagonal cell of radius $400$m, the users are randomly placed in the cell, as illustrated in Fig. \ref{fig:network_topology}. The $C$ BS clusters are positioned equidistant from the center and uniformly spaced at $360^\circ/C$ intervals around the circle. The BS is equipped with $N_t$ transmit antennas and $K$ users are each equipped with $N_r=4$ receive antennas to receive $d_k=4$ data streams. The number of antennas in each cluster is set to $L=128$. To demonstrate the effectiveness of the proposed algorithm, the parameters $N_t$ and $K$ are chosen to satisfy $N_t\geq CKN_r$. The downlink distance-dependent path-loss is simulated by $128.1+37.6 \log _{10} d+\psi$ (in dB), where $d$ is the BS-to-user distance in km, and $\psi$ is the zero-mean Gaussian random variable with $8$dB standard deviation for the shadowing effect. In addition, the Rayleigh fading model is adopted to characterize small-scale fading. The background noise level $\sigma_k^2$ equals $-80$ dBm for all users. The weights $\omega_k$ of all users are set to $1$. Our proposed algorithm uses Armijo backtracking with parameters $\sigma =0.6, p=0.1, t^{0}=1\times 10^{10}$. The results are averaged over 1000 independent realizations.

We benchmark our proposed algorithm against the following algorithms. 
\begin{itemize}
    \item \textbf{WMMSE} \cite{shi2011WMMSE}: The weighted minimum mean square error (WMMSE) method reformulates the WSR maximization as a tractable equivalent problem by introducing auxiliary variables, and then solves it via block coordinate descent over auxiliary and precoding variables iteratively. Given that the constraint \eqref{eq:pgpc} is convex, the optimization over the precoding variable with auxiliary variables fixed remains a convex problem, which can be solved by CVX, a package for solving convex programs \cite{cvx}. Therefore, the WMMSE algorithm can be generalized to per-cluster power constraints case.
    \item \textbf{RWMMSE} \cite{zhao2023rethinking}: RWMMSE is the reduced version of WMMSE in two aspects, in which the dimension of the precoding variable is reduced by the low-dimensional subspace property \cite[Proposition 2]{zhao2023rethinking} and the sum-power constraint is eliminated by the full power property \cite[Proposition 3]{zhao2023rethinking}. However, the elimination of power constraints does not generalize to per-cluster power constraints case directly. In our simulations, we generalize RWMMSE to per-cluster power constraints case by first transforming the problem with Low-Dimensional Subspace Property, and then recovering ${\bV}$ from ${\bX}$ following the per-cluster power constraints.
    \item \textbf{Riemannian conjugate gradient} \cite{sun2024RCG}: The WSR maximization is reformulated on the product manifold \eqref{eq:equality} and solved by the Riemannian conjugate gradient algorithm.
    \item \textbf{EZF} \cite{sun2010ezf}: Eigen zero-forcing (EZF) beamforming is a robust, non-iterative linear precoding strategy in MU-MIMO downlink systems. By performing singular value decomposition (SVD) on individual user channels to align transmissions with dominant eigenmodes, EZF achieves near-optimal spectral efficiency with reduced multi-user interference.
\end{itemize}

\begin{figure}[t]
\centering
\subfloat[$C=4$ clusters]{\includegraphics[width=0.5\linewidth]{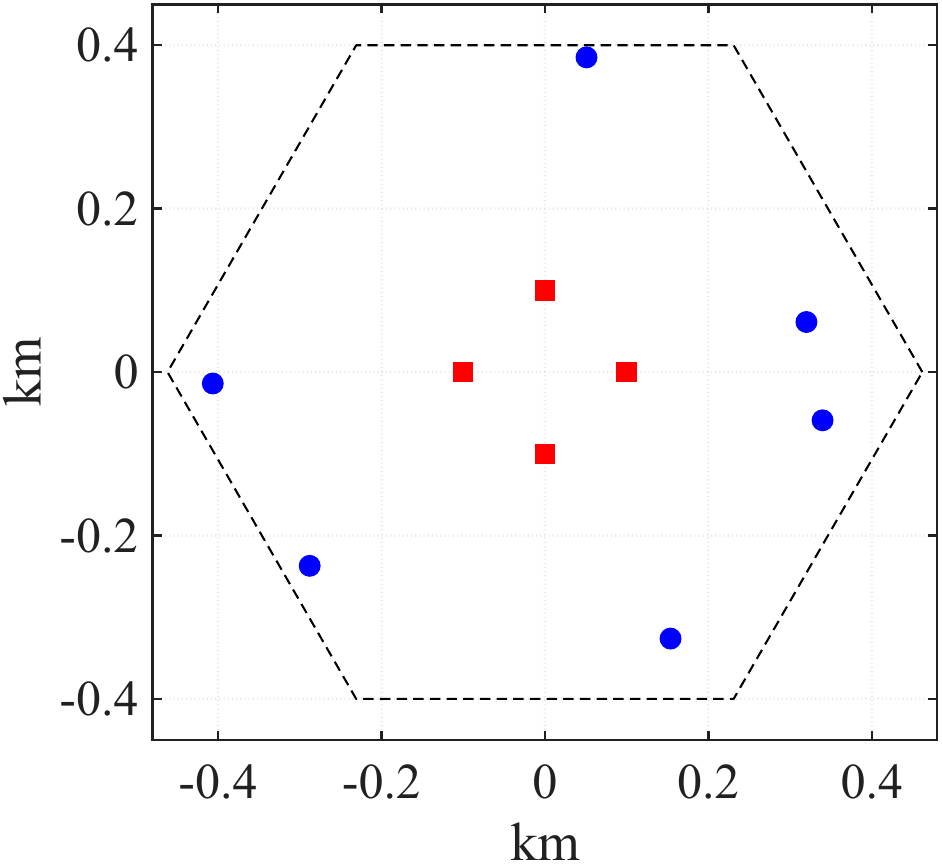}}\hfill
\subfloat[$C=8$ clusters]{\includegraphics[width=0.5\linewidth]{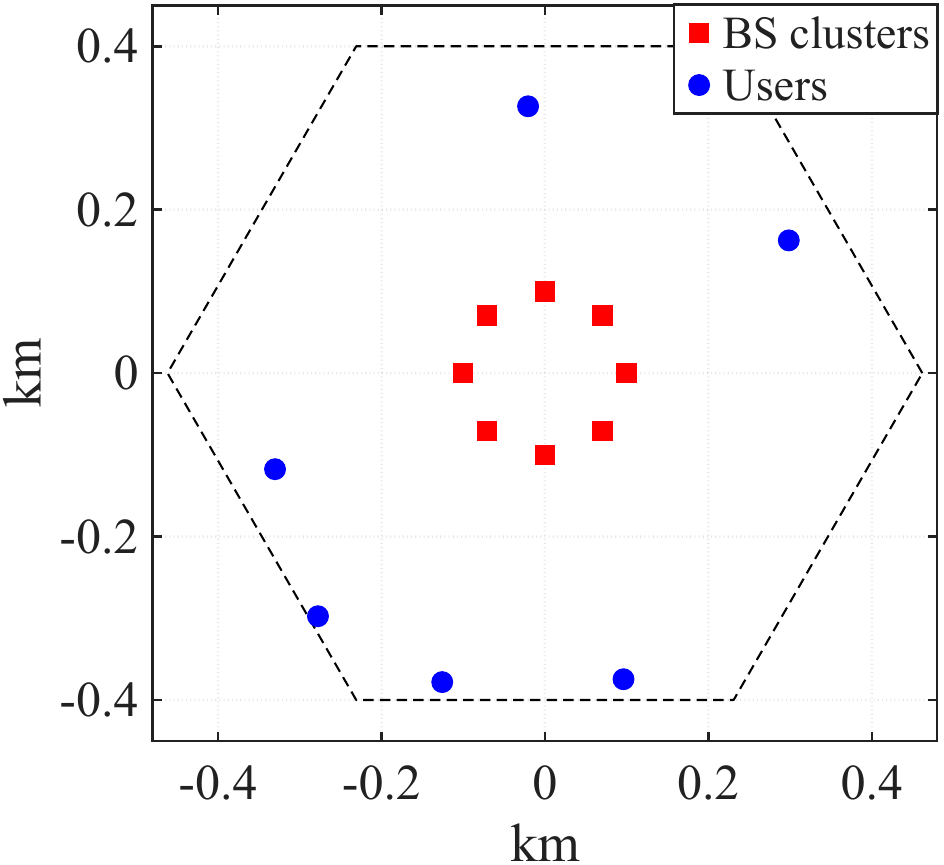}}
\caption{Network topology with hexagonal cell for $C=4$ clusters and $C=8$ clusters.}
\label{fig:network_topology}
\end{figure}

\subsection{Performance Under Clustered Antenna Configurations}

\begin{figure*}[!t]
\centering
\subfloat[when $C=4, \bP=\vecnotation{100,500, 100, 500}$]{
  \centering
  \includegraphics[width=0.47\textwidth]{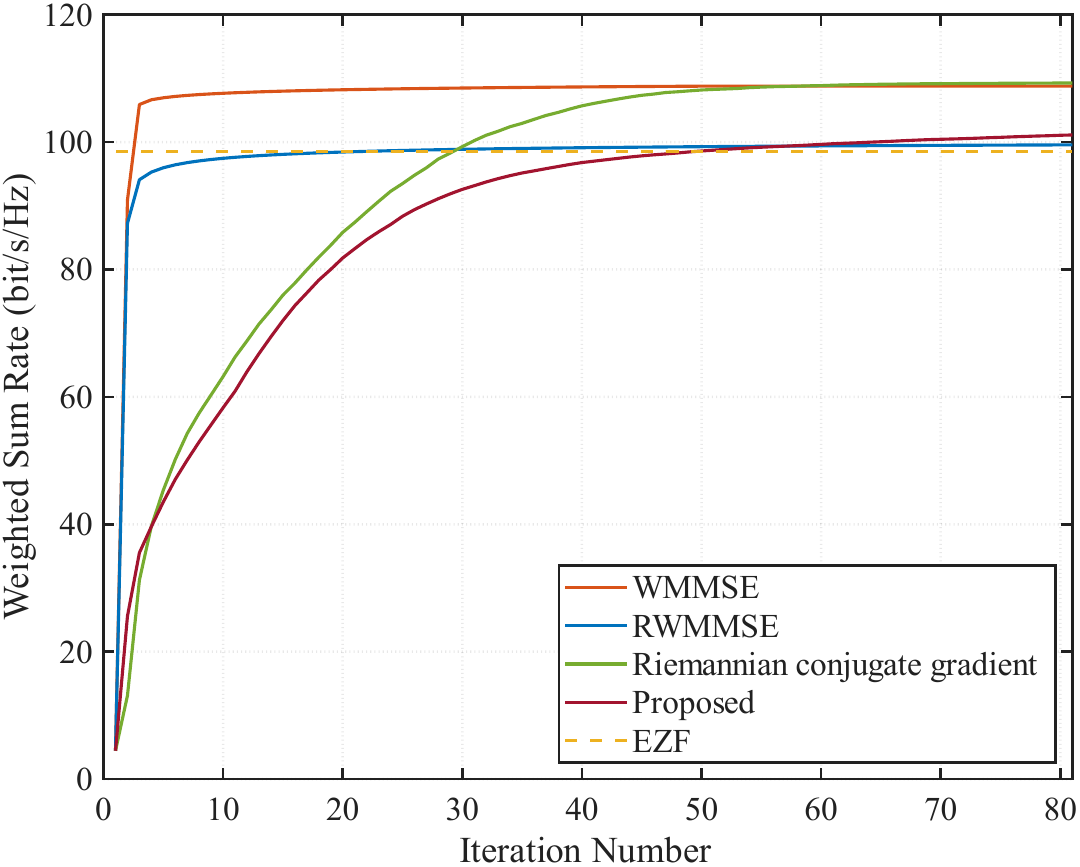}
}
\hfill
\subfloat[when $C=8$, $\bP=\vecnotation{100,500, 100, 500, 100, 500, 100, 500}$]{
  \centering
  \includegraphics[width=0.47\textwidth]{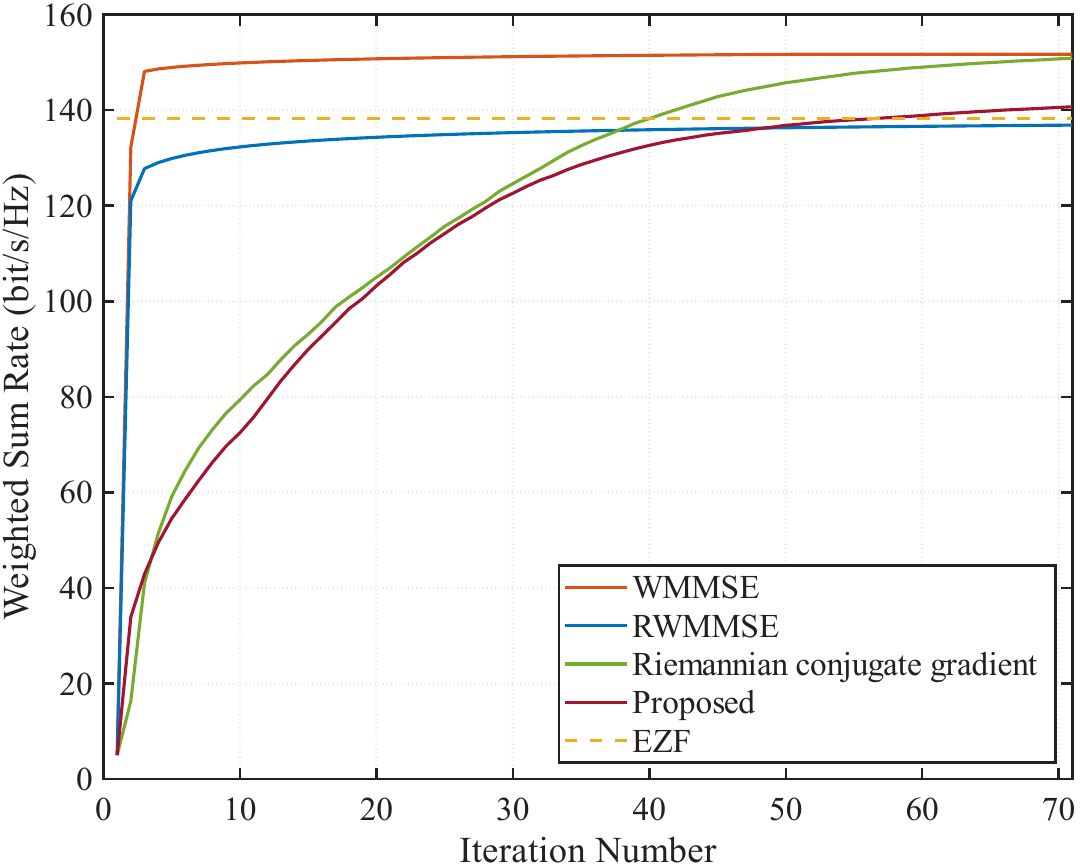}
}
\\[5pt]
\subfloat[when $C=4, \bP=\vecnotation{100,500, 100, 500}$]{
  \centering
  \includegraphics[width=0.47\textwidth]{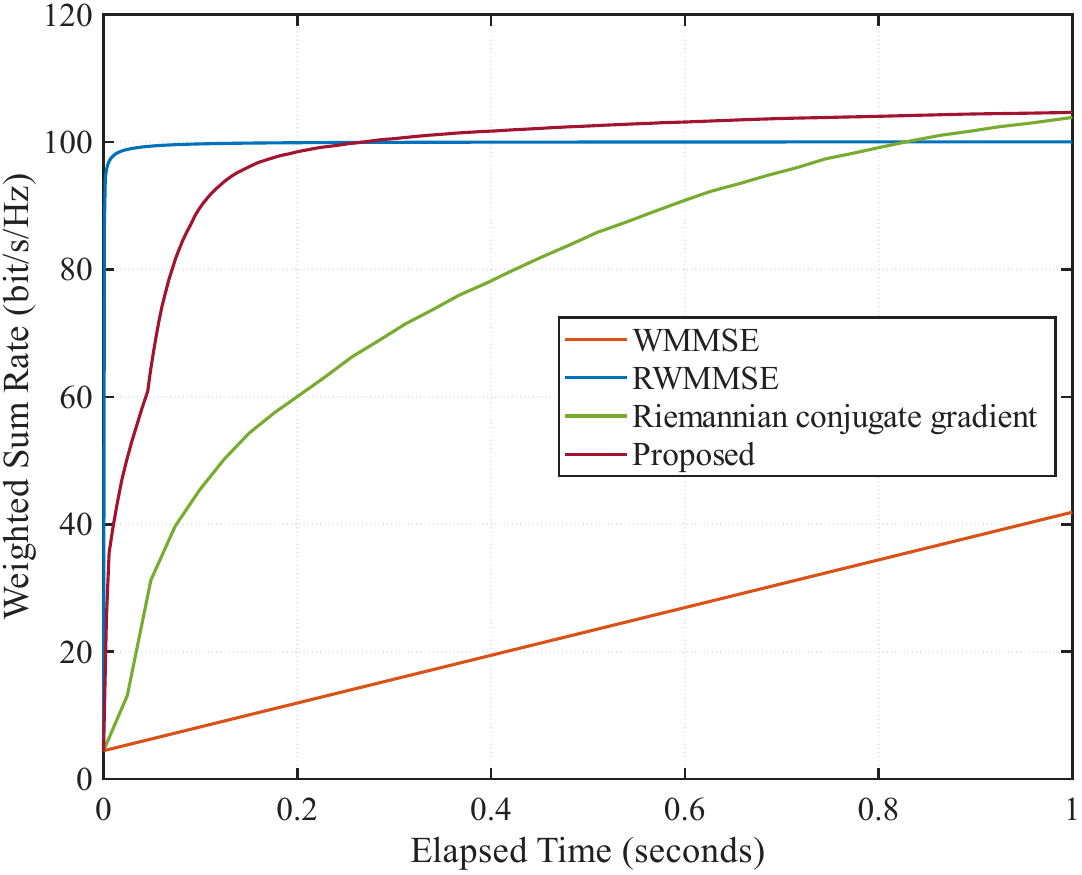}
}
\hfill
\subfloat[when $C=8$, $\bP=\vecnotation{100,500, 100, 500, 100, 500, 100, 500}$]{
  \centering
  \includegraphics[width=0.47\textwidth]{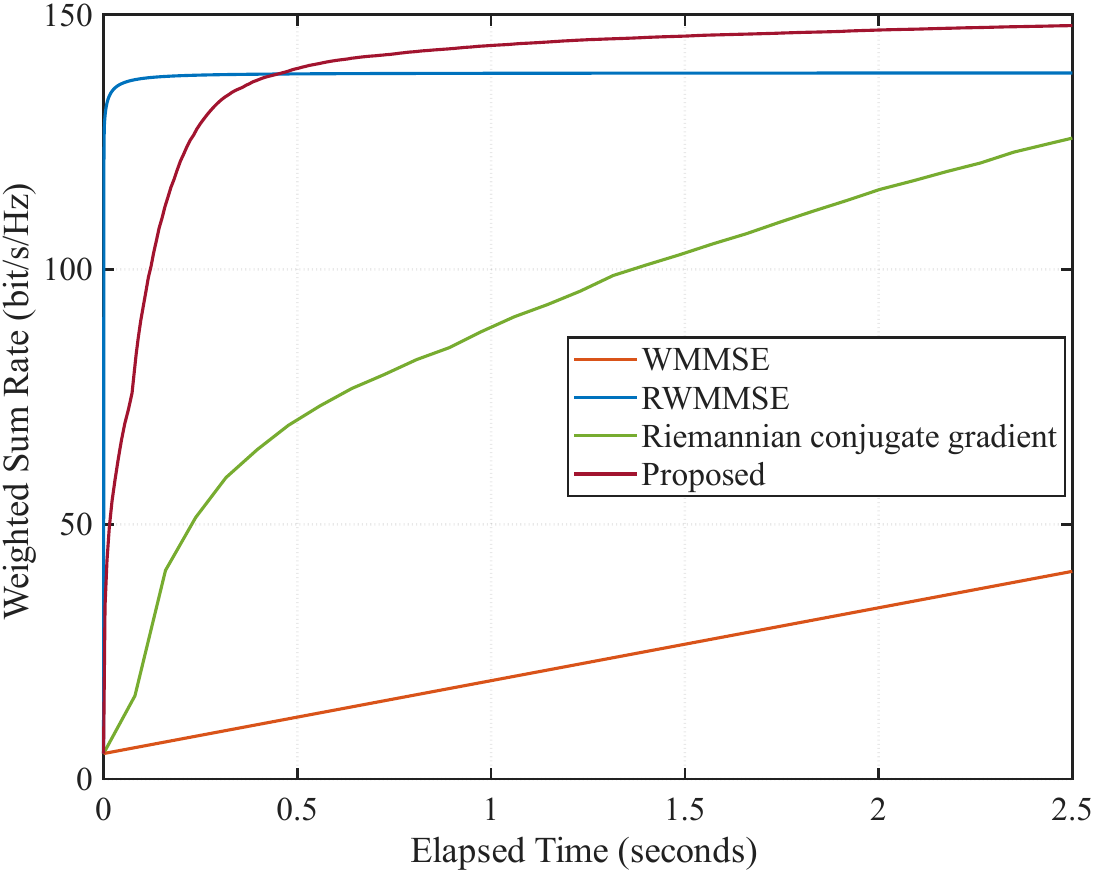}
}
\caption{Cluster performance with $L=128, K=6, N_r=4, d_k=4$. The two figures in the top row show the convergence of WSR in iterations, and the two figures in the bottom row show the convergence in elapsed time.}
\label{fig:cluster_performance}
\end{figure*}

We evaluate the performance of the proposed  algorithm with varying number of clusters $C=\{4,8\}$ and $K=6$ users. The total number of transmit antennas at the BS is $N_t=C\times L$. The power budgets for each cluster are set as follows: for $C=4$, $\bP=\vecnotation{100,500, 100, 500}$; and for $C=8$, $\bP=\vecnotation{100,500, 100, 500, 100, 500, 100, 500}$. As shown in the top row of Fig. \ref{fig:cluster_performance}, our proposed algorithm approaches the same local optimum as WMMSE and Riemannian conjugate gradient, while RWMMSE fails to reach the same optimum. Note that the apparent ``non-convergence'' of our proposed algorithm in the iteration plots is mainly due to the fixed iteration budget in the simulations: because our proposed algorithm needs more iterations to fully settle, it may not reach the final plateau within the displayed number of iterations. Nevertheless, the bottom row of Fig. \ref{fig:cluster_performance} shows that our proposed algorithm is still faster than Riemannian conjugate gradient in terms of elapsed time, since each proposed iteration is significantly cheaper after the low-dimensional reformulation. Moreover, the runtime gap between our proposed algorithm and Riemannian conjugate gradient becomes larger as the number of clusters $C$ increases. This is because the inequality $N_t=C\times L \gg CKN_r$ holds more evidently when $C$ increases, which leads to a smaller size of the reformulated decision variable ${\bX}\in\mathbb{C}^{CKN_r\times Kd_k}$ in our proposed algorithm compared to that in Riemannian conjugate gradient, i.e., ${\bV}\in\mathbb{C}^{N_t\times Kd_k}$. This demonstrates the advantage of our proposed  algorithm in handling per-cluster power constraints with a large number of antenna clusters $C$.

To further evaluate the statistical performance of all algorithms, we plot the cumulative distribution function (CDF) of the WSR achieved by each algorithm. For a fair comparison under the same computational budget, each algorithm is given a fixed running time of 2.5 seconds with $C=8$ clusters. Note that WMMSE is excluded from this comparison because its high per-iteration complexity (due to repeated matrix inversions and CVX solver calls) prevents it from completing sufficient iterations within the 2.5-second budget to achieve meaningful convergence. As shown in Fig. \ref{fig:cluster_cdf}, the proposed  algorithm achieves the highest WSR over 90\% of the CDF range, outperforming RWMMSE and Riemannian conjugate gradient. This demonstrates that, within a fixed time budget, our proposed algorithm is able to find better solutions due to its lower per-iteration complexity enabled by the low-dimensional reformulation. The CDF plot further confirms the practical advantage of our proposed algorithm for real-time or time-critical massive MIMO systems under per-cluster power constraints.


\begin{figure}[h]
\centering
\includegraphics[width=0.47\textwidth]{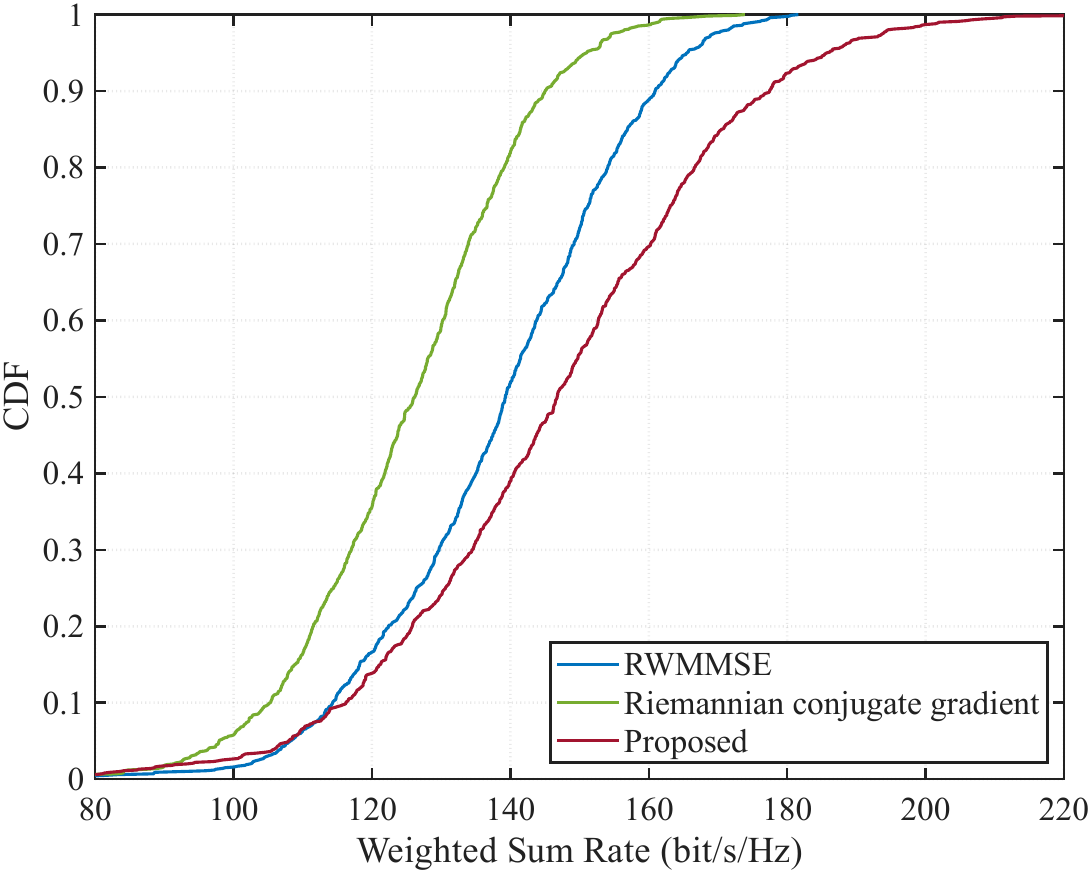}
\caption{CDF of the WSR achieved by different algorithms under cluster configurations with $C=8$, given a fixed running time of 2.5 seconds.}
\label{fig:cluster_cdf}
\end{figure}

Next, we compare the convergence time of our proposed algorithm and Riemannian conjugate gradient under different number of users $K$ and different number of antennas per cluster $L$, where the convergence time is defined as the elapsed time until the WSR change between two consecutive iterations satisfies $\left|\mathrm{WSR}^{(\ell)}-\mathrm{WSR}^{(\ell-1)}\right|<10^{-3}$. As shown in Fig. \ref{fig:cluster_time_K}, the convergence time of our proposed algorithm is lower than that of Riemannian conjugate gradient for all considered values of $K$. our proposed algorithm is more efficient than Riemannian conjugate gradient especially when the number of users is small, which follows from the smaller size of the reformulated decision variable ${\bX}\in\mathbb{C}^{CKN_r\times Kd_k}$ in our proposed algorithm compared to that in Riemannian conjugate gradient, i.e., ${\bV}\in\mathbb{C}^{N_t\times Kd_k}$, when $N_t$ is much larger than $CKN_r$. As shown in Fig. \ref{fig:cluster_time_L}, when fixing $K=6$, the convergence time of our proposed algorithm is lower than that of Riemannian conjugate gradient for all considered values of $L$. The convergence time of Riemannian conjugate gradient increases significantly as $L$ increases, while that of our proposed algorithm increases slightly. This demonstrates that our proposed algorithm is more efficient than Riemannian conjugate gradient especially when the number of antennas per cluster $L$ is large.

\begin{figure}[t]
\centering
\includegraphics[width=0.47\textwidth]{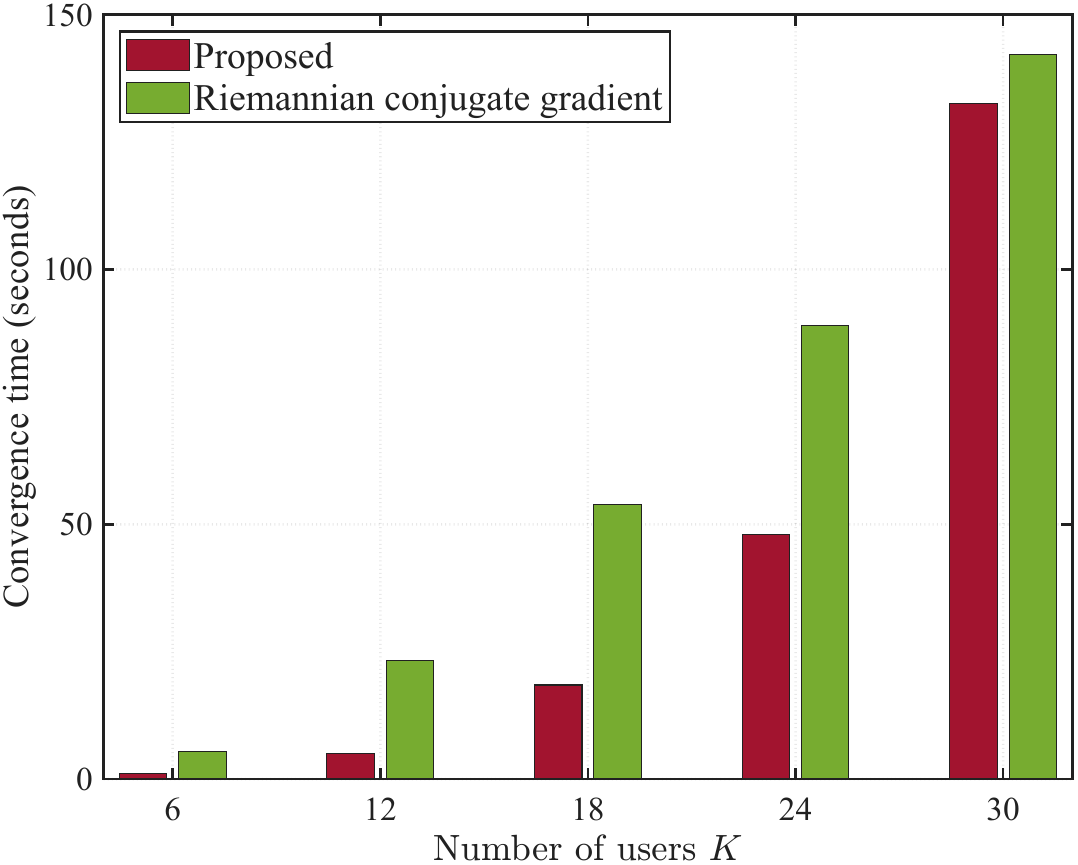}
\caption{Convergence time (elapsed time until $\left|\mathrm{WSR}^{(\ell)}-\mathrm{WSR}^{(\ell-1)}\right|<10^{-3}$) under per-cluster power constraints with $C=8, L=128, N_r=2, d_k=2$ and varying $K$.}
\label{fig:cluster_time_K}
\end{figure}

\begin{figure}[h]
\centering
\includegraphics[width=0.47\textwidth]{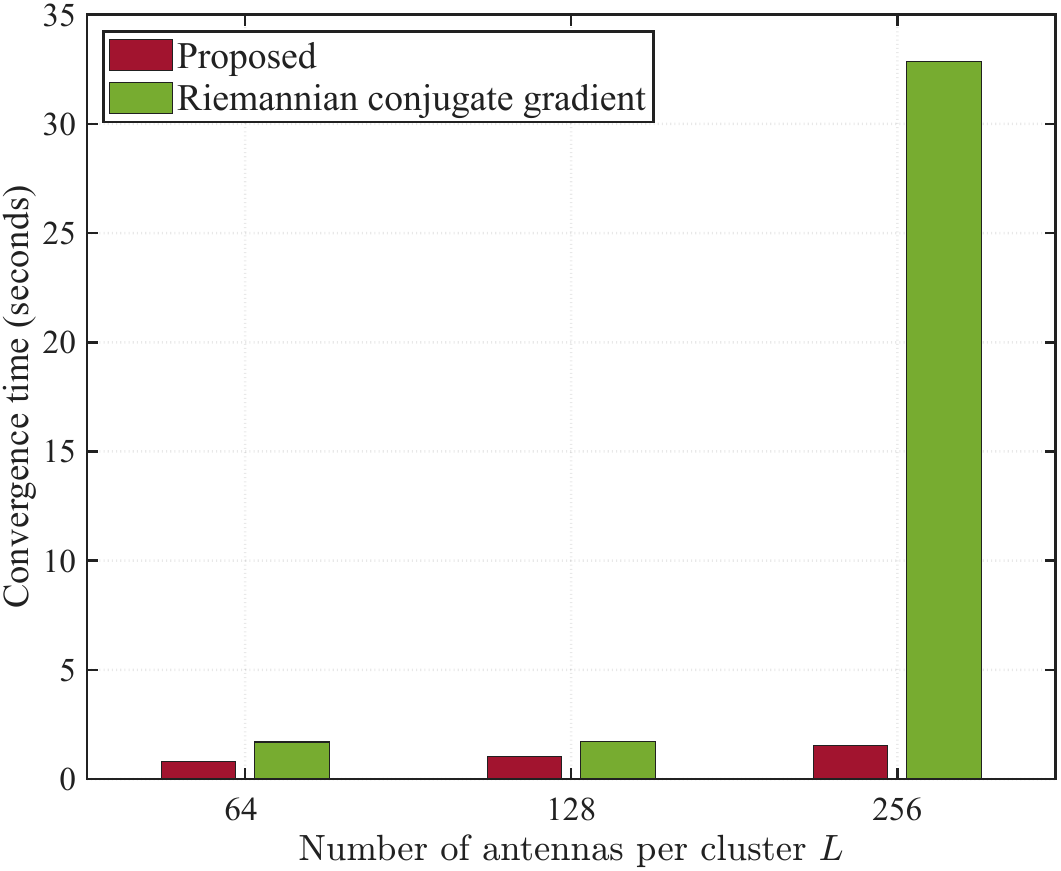}
\caption{Convergence time (elapsed time until $\left|\mathrm{WSR}^{(\ell)}-\mathrm{WSR}^{(\ell-1)}\right|<10^{-3}$) under per-cluster power constraints with $C=8, K=6, N_r=2, d_k=2$ and varying $L$.}
\label{fig:cluster_time_L}
\end{figure}

\section{Conclusion} \label{sec:conclusion}
This paper studied WSR maximization for massive MIMO downlink beamforming under per-cluster power constraints, a practically motivated setting arising from clustered-antenna architectures and cluster-wise power budgets. By extending the low-dimensional subspace property beyond the classical sum-power case, we proved that any nontrivial stationary solution admits a reduced representation where the beamformer of each antenna cluster lies in the subspace spanned by the corresponding effective channels. This yields a principled dimensionality reduction of the design variable from $N_t$ to $CKN_r$, while preserving the stationarity structure of the original problem.

Leveraging this structure, we reformulated the per-cluster power constraint-constrained WSR problem as an unconstrained optimization over a product of ellipsoid manifolds induced by cluster-wise power equalities. We derived the core Riemannian components on this product manifold, including tangent spaces, inner products, orthogonal projections, retractions, and vector transports, and developed an algorithm equipped with Hestenes--Stiefel updates and Armijo backtracking line search.

Numerical results under per-cluster power constraint scenarios demonstrated that our proposed algorithm reaches the same local optima as standard baselines such as WMMSE and Riemannian conjugate gradient, while substantially reducing elapsed time due to the smaller per-iteration computational burden enabled by the reduced formulation. The runtime gains become more pronounced as the number of antenna clusters $C$ grows and as the per-cluster antenna count increases; moreover, within a fixed time budget, our proposed algorithm achieves favorable WSR statistics in terms of the CDF compared to benchmark methods. These results highlight the effectiveness and scalability of the proposed low-dimensional manifold framework for large-scale massive MIMO precoding with structured power constraints.

\appendices
\renewcommand{\thesubsection}{\thesection.\arabic{subsection}}
\makeatletter
\renewcommand{\thesubsectiondis}{\thesubsection}
\makeatother

\section{Proof of Proposition \ref{prop:power tight}}\label{appendix:power_tight}
By defining the negative objective function of problem \eqref{eq:original} as $f({\bV}) = -\sum_{k=1}^K \omega_k R_k$. The KKT conditions of problem \eqref{eq:original} state that for any stationary point ${\bV}^*$, there exists non-negative multipliers $\mu_c^* \ge 0, c=1,\ldots,C$, such that the following conditions hold:
\begin{subequations}\label{eq:kkt}
\begin{align}
&\nabla_{\bV^{(c)}} f({\bV}^*) + \mu_c^* \, {(\bV^{(c)})^*} = \mathbf{0}, \forall c \label{eq:kkt_stationarity}\\
&\mu_c^*\, \left(\|({\bV}^{(c)})^*\|_F^2 - P_c\right) = 0, \quad \forall c, \label{eq:kkt_compslack}\\
&\|({\bV}^{(c)})^*\|_F^2\leq P_c,  \forall c, \label{eq:kkt_primal}\\
&\mu_c^* \ge 0, \quad \forall c, \label{eq:kkt_dual}
\end{align}
\end{subequations}
where \eqref{eq:kkt_stationarity} is the first-order stationarity condition with respect to ${\bV}^{(c)}$, \eqref{eq:kkt_compslack} is the complementary slackness condition, \eqref{eq:kkt_primal} is the primal feasibility condition, and \eqref{eq:kkt_dual} is the dual feasibility condition.

Next, we will prove $\mu_c > 0$ for all $c$ by contradiction. Note that we will only prove for an arbitrary $c$, since the proof is similar for all $c$. Assume, by contradiction, that $\mu_c = 0$. Accordingly, the stationarity condition \eqref{eq:kkt_stationarity} and the complementary slackness condition \eqref{eq:kkt_compslack} reduce to
\begin{subequations}
\begin{align}
&\nabla_{{\bV}^{(c)}} f({\bV}^*) = \mathbf{0}, \label{eq:reduced_stationarity}\\
&\|({\bV}^{(c)})^*\|_F^2 < P_c. \label{eq:reduced_compslack}
\end{align}
\end{subequations}
Here we pick an arbitrary user $k$ and the channel matrix from $c$th cluster to user $k$ is $\bH_{k}^{(c)}\in\mathbb{C}^{N_r \times L}$. Since $\bH^{(c)}$ has full row rank and $L\ge N_r$, we have $\operatorname{rank}(\bH_{k}^{(c)})=N_r$. According to the rank-nullity theorem, the dimension of the null space of $\bH_{k}^{(c)}$ is
$\operatorname{dim}(\operatorname{null}(\bH_{k}^{(c)})) = L - N_r > 0$. Therefore, there exists a non-zero vector $\boldsymbol{\nu} \in \mathbb{C}^{L}$ such that $\bH_{k}^{(c)} \boldsymbol{\nu} = \mathbf{0}$.

We write the channel from $c$th cluster without user $k$ as
\begin{equation}
\bH_{-k}^{(c)} = \begin{bmatrix}\bH_{1}^{(c)} \\ \vdots \\ \bH_{k-1}^{(c)} \\ \bH_{k+1}^{(c)} \\ \vdots \\ \bH_{K}^{(c)}\end{bmatrix}\in\mathbb{C}^{(K-1)N_r \times L}.
\end{equation}
Since $\bH^{(c)}$ has full row rank, its submatrix $\bH_{-k}^{(c)}$ also has full row rank. With $L\ge KN_r>(K-1)N_r$, we have $\operatorname{rank}(\bH_{-k}^{(c)})=(K-1)N_r$. Similarly, we have $\operatorname{dim}(\operatorname{null}(\bH_{-k}^{(c)})) = L - (K-1)N_r > 0.$
That is, there exists a non-zero vector $\bu \in \mathbb{C}^{L}$ such that $\bH_{-k}^{(c)} \bu = \mathbf{0}$.

We define the two non-trivial subspaces
\begin{equation}
\mathcal{U} \triangleq \operatorname{null}(\bH_{-k}^{(c)}), 
\end{equation}
\begin{equation}
\mathcal{V} \triangleq \operatorname{null}(\bH_{k}^{(c)}). 
\end{equation}
We make two mutually exclusive cases as follows. 

\begin{enumerate}
  \item \textbf{Case 1}: $\mathcal{U} \subseteq \mathcal{V}$. In this case, every vector in $\mathcal{U}$ is also in $\mathcal{V}$. Hence, for any $\bu \in \mathcal{U}$, we have
  \begin{equation}
  \bH_{-k}^{(c)}\bu=\mathbf{0}\quad\text{and}\quad \bH_{k}^{(c)}\bu=\mathbf{0}.
  \end{equation}

  \item \textbf{Case 2}: $\mathcal{U} \nsubseteq \mathcal{V}$. In this case, there exists at least one vector $\bu \in \mathcal{U}$ such that $\bu \notin \mathcal{V}$, i.e.,
  \begin{equation}
  \bH_{-k}^{(c)}\bu=\mathbf{0}\quad\text{and}\quad \bH_{k}^{(c)}\bu \neq \mathbf{0}.
  \end{equation}
\end{enumerate}

Next, we show that \textbf{Case 1} is impossible. The channel matrix $\bH_k^{(c)}$ is composed of pathloss, shadowing, and Rayleigh fading components, i.e.,$\bH_k^{(c)} = \sqrt{\beta_k^{(c)}} \bE_k^{(c)}$, where $\beta_k^{(c)}$ is the large-scale fading coefficient and $\bE_k^{(c)}$ is the Rayleigh fading matrix whose entries are independently drawn from the complex Gaussian distribution $\mathcal{CN}(0,1)$. Since $\beta_k^{(c)} > 0$, we have $\operatorname{null}(\bH_{k}^{(c)}) = \operatorname{null}(\bE_{k}^{(c)})$. Let $\bB \in \mathbb{C}^{L \times \operatorname{dim}(\mathcal{U})}$ be a non-zero basis matrix for $\mathcal{U}$, we have $\bE_{k}^{(c)} \bB = \mathbf{0}$. The probability of \textbf{Case 1} occurring can be equivalently expressed as
\begin{equation}
\Pr(\mathcal{U} \subseteq \mathcal{V}) = \Pr(\mathcal{U} \subseteq \operatorname{null}(\bE_{k}^{(c)})) = \Pr(\bE_{k}^{(c)} \bB = \mathbf{0}).
\end{equation}
Since the entries of $\bE_{k}^{(c)}$ are independently drawn from a continuous distribution, the probability that $\bE_{k}^{(c)} \bB = \mathbf{0}$ holds is zero.

Thus, \textbf{Case 1} never occurs, so it suffices to consider \textbf{Case 2}. Therefore, there exists a non-zero vector $\bu \in \mathbb{C}^{L}$ such that
\begin{equation}
\bH_{-k}^{(c)}\bu=\mathbf{0}\quad\text{and}\quad \bH_{k}^{(c)}\bu \neq \mathbf{0}. \label{eq:null}
\end{equation}
Based on this $\bu$, we define
\begin{equation}
  \Delta {\bV}^{(c)}_k = \bu \be^\hh,
\end{equation}
as a perturbation to $({\bV}_k^{(c)})^*$, where $\be\in\mathbb{C}^{d_k}$ is an arbitrary non-zero vector. For user $k^\prime \neq k$ and cluster $c^\prime \neq c$, the perturbations are set to zero, i.e., $\Delta {\bV}^{(c)}_{k^\prime} = \mathbf{0}$ and $\Delta {\bV}^{(c^\prime)} = \mathbf{0}$. The perturbed beamforming matrix is defined as
\begin{equation}
{\bV}(\epsilon) = {\bV}^* + \epsilon \Delta {\bV}, \label{eq:V_eps}
\end{equation}
where $\epsilon>0$ is a small constant such that the complementary slackness condition \eqref{eq:reduced_compslack} remains slack, i.e., $\|{\bV}^{(c)}+\epsilon \Delta {\bV}^{(c)}\|_F^2 < P_c$.

Given \eqref{eq:null}, we will analyze the rate change of each user when the perturbation \eqref{eq:V_eps} is applied. On one hand, for user $k^\prime \neq k$, its received signal is
\begin{equation}
\begin{aligned}
\bH_{k^\prime} {\bV}(\epsilon) &= \bH_{k^\prime} \left({\bV}^* + \epsilon \Delta {\bV}\right) \\
&= \bH_{k^\prime} {\bV}^* + \epsilon \bH_{k^\prime} \Delta {\bV},
\end{aligned}
\end{equation}
where $\epsilon \bH_{k^\prime} \Delta {\bV}$ can be further expressed as
\begin{equation}
\begin{aligned}
\epsilon \bH_{k^\prime} \Delta {\bV} 
&= \epsilon \bH_{k^\prime}^{(c)} \Delta {\bV}^{(c)} \\
&= \epsilon \bH_{k^\prime}^{(c)} \begin{bmatrix}\mathbf{0} & \cdots & \Delta {\bV}^{(c)}_k & \cdots & \mathbf{0}\end{bmatrix} \\
&= \epsilon \bH_{k^\prime}^{(c)} \begin{bmatrix}\mathbf{0} & \cdots & \bu \be^\hh & \cdots & \mathbf{0}\end{bmatrix} \\
&= \begin{bmatrix}\mathbf{0} & \cdots & \epsilon \bH_{k^\prime}^{(c)} \bu \be^\hh & \cdots & \mathbf{0}\end{bmatrix} \\
&= \mathbf{0},
\end{aligned}
\end{equation}
where the last equality follows from \eqref{eq:null}. Therefore, the received signal of user $k^\prime \neq k$ remains unchanged after the perturbation and the achievable rate $R_{k^\prime}$ also remains unchanged. On the other hand, for user $k$, its received signal is
\begin{equation}
\begin{aligned}
\bH_{k} {\bV}(\epsilon) &= \bH_{k} \left({\bV}^* + \epsilon \Delta {\bV}\right) \\
&= \bH_{k} {\bV}^* + \epsilon \bH_{k} \Delta {\bV},
\end{aligned}
\end{equation}
where $\epsilon \bH_{k} \Delta {\bV}$ can be further expressed as
\begin{equation}
\begin{aligned}
\epsilon \bH_{k} \Delta {\bV} 
&= \epsilon \bH_{k}^{(c)} \Delta {\bV}^{(c)} \\
&= \epsilon \bH_{k}^{(c)} \begin{bmatrix}\mathbf{0} & \cdots & \Delta {\bV}^{(c)}_k & \cdots & \mathbf{0}\end{bmatrix} \\
&= \epsilon \bH_{k}^{(c)} \begin{bmatrix}\mathbf{0} & \cdots & \bu \be^\hh & \cdots & \mathbf{0}\end{bmatrix} \\
&= \begin{bmatrix}\mathbf{0} & \cdots & \epsilon \bH_{k}^{(c)} \bu \be^\hh & \cdots & \mathbf{0}\end{bmatrix} ,
\end{aligned}
\end{equation}
where $\epsilon \bH_{k}^{(c)} \bu \be^\hh \neq \mathbf{0}$ follows from \eqref{eq:null}. The last line indicates that the desired signal of user $k$ is enhanced after the perturbation, while the interference-plus-noise power remains unchanged. Therefore, the achievable rate $R_{k}$ is strictly increased after the perturbation. Therefore, the objective function value $f({\bV}(\epsilon))$ is strictly decreased after the perturbation, i.e.,
\begin{equation}
f({\bV}(\epsilon)) < f({\bV}^*).
\end{equation}
This contradicts the reduced stationarity condition \eqref{eq:reduced_stationarity}. Therefore, our initial assumption that $\mu_c = 0$ must be false, and the complementary slackness condition \eqref{eq:kkt_compslack} implies that the power constraints are tight at any stationary point, i.e., $\|({\bV}^{(c)})^*\|_F^2 = P_c$ for all $c$.

\section{Proof of Proposition \ref{prop:reform_equiv}}\label{appendix:reform_equiv}
\subsection{Cluster-wise low-dimensional subspace property via explicit gradients}

For each user $k$,
\begin{align}
R_k
&= \log\det\left(\bI+\bV_{k}^\hh\bH_{k}^\hh\bF_{k}^{-1}\bH_{k}\bV_{k}\right)\notag\\
&= \log\det(\bS_k)-\log\det(\bF_k),
\label{eq:Rk_logdet_diff}
\end{align}
with
\begin{align}
\bF_k &= \sigma_k^2\bI+\sum_{j=1,j\neq k}^K\bH_k\bV_j\bV_j^\hh\bH_k^\hh., \label{eq:Fk_def}\\
\bS_k &= \sigma_k^2\bI + \sum_{j=1}^{K}\bH_k\bV_j\bV_j^\hh\bH_k^\hh. \label{eq:Sk_def}
\end{align}
Recall the clustered structure $\bH_k=[\bH_k^{(1)},\dots,\bH_k^{(C)}]$ and
$\bV_j=[(\bV_j^{(1)})^{\top},\dots,(\bV_j^{(C)})^{\top}]^{\top}$, so
\begin{equation}
\bH_k\bV_j=\sum_{c=1}^{C}\bH_k^{(c)}\bV_j^{(c)}.
\end{equation}
Fix any $k$ and any beamformer block $\bV_j^{(c)}$ (user $j$, cluster $c$). Let
\begin{equation}
\bS_{k,j}\triangleq \bH_k\bV_j\bV_j^\hh\bH_k^\hh.
\end{equation}
Then $\bS_k=\sigma_k^2\mathbf I+\sum_{j=1}^{K} \bS_{k,j}$ and $\bF_k=\sigma_k^2\mathbf I+\sum_{j=1,j\neq k}^{K}\bS_{k,j}$.
Using the matrix differential identity $\mathrm{d} \left(\log\det(\mathbf A)\right)=\operatorname{tr}\left(\mathbf A^{-1}\mathrm{d}\mathbf A\right)$, we obtain the differential of $R_k$ with respect to $\bV_i^{(c)}$ as
\begin{equation}
\begin{aligned}
\mathrm{d}R_k &= \operatorname{tr} \left(\bS_k^{-1}\, \mathrm{d}\bS_k\right) - \operatorname{tr} \left(\bF_k^{-1}\, \mathrm{d}\bF_k\right) \\
&= \operatorname{tr} \left(\bM_{k,i}\, \mathrm{d}\bS_{k,i}\right),
\label{eq:dRk_master}
\end{aligned}
\end{equation}
where
\begin{equation}
\bM_{k,i}=
\begin{cases}
\bS_k^{-1}, & i=k,\\[2mm]
\bS_k^{-1}-\bF_k^{-1}, & i\neq k.
\end{cases}
\label{eq:Mkj_def}
\end{equation}
Next we compute $\mathrm{d}\bS_{k,i}$ as follows
\begin{align}
\mathrm{d}\bS_{k,i}
&= \mathrm{d}(\bH_k\bV_i\bV_i^\hh\bH_k^\hh) \\
&= \bH_k^{(c)}\, \mathrm{d}\bV_i^{(c)}\, \bV_i^\hh\bH_k^\hh
\;+\; \bH_k\bV_i\, \mathrm{d}(\bV_i^{(c)})^\hh\, (\bH_k^{(c)})^\hh.
\label{eq:dSkj}
\end{align}
Plugging \eqref{eq:dSkj} into \eqref{eq:dRk_master} yields
\begin{multline}
\mathrm{d}R_k = \operatorname{tr} \left(\bM_{k,i}\, \bH_k^{(c)}\, \mathrm{d}\bV_i^{(c)}\, \bV_i^\hh\bH_k^\hh\right) \\
\quad + \operatorname{tr} \left(\bM_{k,i}\, \bH_k\bV_i\, \mathrm{d}(\bV_i^{(c)})^\hh\, (\bH_k^{(c)})^\hh\right)
\label{eq:dRk_expanded}
\end{multline}
Since $R_k$ is real-valued, by Wirtinger gradient, the gradient with respect to the complex matrix variable $\bV_i^{(c)}$ can be identified from \eqref{eq:dRk_expanded} as
\begin{equation}
\nabla_{\bV_i^{(c)}} R_k
= (\bH_k^{(c)})^\hh\, \bM_{k,i}\bH_k\, \bV_i.
\label{eq:gradVjc_Rk}
\end{equation}
\eqref{eq:gradVjc_Rk} shows that $\nabla_{\bV_i^{(c)}} R_k$ lies in the column space of $(\bH_k^{(c)})^\hh$, i.e.,
\begin{equation}
\nabla_{\bV_i^{(c)}} R_k \in \operatorname{Col}((\bH_k^{(c)})^\hh).
\label{eq:gradRk_in_range}
\end{equation}
Therefore, the gradient of $f(\bV)$ with respect to $\bV_i^{(c)}$ is
\begin{equation}
\begin{aligned}
\nabla_{\bV_i^{(c)}} f(\bV) &= -\sum_{k=1}^{K} \omega_k \nabla_{\bV_i^{(c)}} R_k \\
&= -\sum_{k=1}^{K} \omega_k (\bH_k^{(c)})^\hh\, \bM_{k,i}\bH_k\, \bV_i,
\label{eq:gradVjc_f}
\end{aligned}
\end{equation}
where the $k$th summand belongs to $\operatorname{Col}(\bH_k^{(c)})^\hh$. Therefore, the summation $\nabla_{\bV_i^{(c)}} f(\bV)$ belongs to the span:
\begin{equation}
\!\!\!\!\nabla_{\bV_i^{(c)}} f(\bV) \!\in\! \operatorname{span}\left\{\! \operatorname{Col}((\bH_1^{(c)})^\hh), \ldots, \operatorname{Col}((\bH_K^{(c)})^\hh) \!\right\}\!.
\label{eq:grad_span}
\end{equation}
Since $(\bH^{(c)})^\hh = [(\bH_1^{(c)})^\hh, \dots, (\bH_K^{(c)})^\hh]$, its column space can be expressed as
\begin{multline}
\operatorname{Col}((\bH^{(c)})^\hh)=\\
\operatorname{span}\left\{ \operatorname{Col}((\bH_1^{(c)})^\hh), \dots, \operatorname{Col}((\bH_K^{(c)})^\hh) \right\} .
\label{eq:span_eq_range}
\end{multline}
Combining \eqref{eq:grad_span} and \eqref{eq:span_eq_range}, we can conclude that $\nabla_{\bV_i^{(c)}} f(\bV) \in \operatorname{Col}((\bH^{(c)})^\hh)$. We now obtain the gradient with respect to the entire cluster block $\bV^{(c)}$ as
\begin{equation}
\nabla_{\bV^{(c)}} f(\bV) = \left[\nabla_{\bV_1^{(c)}} f(\bV), \ldots, \nabla_{\bV_K^{(c)}} f(\bV)\right],
\label{eq:gradVc_f}
\end{equation}
where each block column lies in $\operatorname{Col}((\bH^{(c)})^\hh)$. Therefore, $\nabla_{\bV^{(c)}} f(\bV) \in \operatorname{Col}((\bH^{(c)})^\hh)$.

According to Appendix \ref{appendix:power_tight}, for any stationary point $\bV^*$ of problem \eqref{eq:original}, we have $\mu_c^* > 0$ for all $c$. The stationarity condition \eqref{eq:kkt_stationarity} can be rearranged as
\begin{equation}
 ({\bV}^{(c)})^* = -\frac{1}{\mu_c^*} \nabla_{\bV^{(c)}} f({\bV}^*) \in \operatorname{Col}((\bH^{(c)})^\hh), \forall c. \label{eq:stationarity_rearranged}
\end{equation}
Hence, there exists $\bX^{(c)}$ such that
\begin{equation}
({\bV}^{(c)})^* = (\bH^{(c)})^\hh \bX^{(c)}, \forall c. \label{eq:Vc_transform}
\end{equation}
This result justifies the substitution \eqref{eq:Vkc_transform} used in the reformulation.
\vspace*{-5pt}
\subsection{Rate equivalence and the construction of $\bG_k$}

We now show that the reformulated rate $\hat{R}_k$ in \eqref{eq:rate_reform} is equivalent to the original rate $R_k$ in \eqref{eq:rate} under the substitution $\bV^{(c)}_k = (\bH^{(c)})^\hh\bX^{(c)}_k$.

Recall that the original rate is given by
\begin{equation}
R_k = \log\det\left(\bI+\bV_{k}^\hh\bH_{k}^\hh\bF_{k}^{-1}\bH_{k}\bV_{k}\right),
\end{equation}
where $\bF_k=\sigma_k^2\bI+\sum_{j=1,j\neq k}^K\bH_k\bV_j\bV_j^\hh\bH_k^\hh$.

We first derive the expression for $\bH_k\bV_j$ using the clustered structure. Since $\bH_k = [\bH_k^{(1)}, \ldots, \bH_k^{(C)}]$ and $\bV_j = [(\bV_j^{(1)})^\top, \ldots, (\bV_j^{(C)})^\top]^\top$, we have
\begin{equation}
\bH_k\bV_j = \sum_{c=1}^{C} \bH_k^{(c)} \bV_j^{(c)}.
\end{equation}
Substituting $\bV_j^{(c)} = (\bH^{(c)})^\hh \bX_j^{(c)}$ yields
\begin{equation}
\begin{aligned}
\bH_k\bV_j &= \sum_{c=1}^{C} \bH_k^{(c)} (\bH^{(c)})^\hh \bX_j^{(c)} \\
&= \bG_k \bX_j, \label{eq:HkVj_to_GkXj}
\end{aligned}
\end{equation}
where we define the effective channel matrix
\begin{equation}
\bG_k = \left[\bH_k^{(1)}(\bH^{(1)})^\hh, \ldots, \bH_k^{(C)}(\bH^{(C)})^\hh\right] \in \mathbb{C}^{N_r \times CN_r}.
\label{eq:Gk_construction}
\end{equation}

Using \eqref{eq:HkVj_to_GkXj}, the interference-plus-noise covariance matrix becomes
\begin{equation}
\begin{aligned}
\bF_k &= \sigma_k^2\bI + \sum_{j=1,j\neq k}^K \bH_k\bV_j\bV_j^\hh\bH_k^\hh \\
&= \sigma_k^2\bI + \sum_{j=1,j\neq k}^K \bG_k\bX_j\bX_j^\hh\bG_k^\hh. \label{eq:Fk_reform}
\end{aligned}
\end{equation}
Substituting \eqref{eq:HkVj_to_GkXj} and \eqref{eq:Fk_reform} into the original rate expression, we obtain
\begin{equation}
\begin{aligned}
R_k &= \log\det\left(\bI+\bV_{k}^\hh\bH_{k}^\hh\bF_{k}^{-1}\bH_{k}\bV_{k}\right) \\
&= \log\det\left(\bI + \bX_k^\hh \bG_k^\hh \bF_k^{-1} \bG_k \bX_k\right) \\
&= \hat{R}_k. \label{eq:rate_equivalence}
\end{aligned}
\end{equation}
This establishes that $R_k = \hat{R}_k$ under the substitution $\bV^{(c)}_k = (\bH^{(c)})^\hh\bX^{(c)}_k$, confirming the equivalence between \eqref{eq:equality} and \eqref{eq:reform ellipse}.

\section{Proof of Proposition \ref{prop:ellipsoid}}\label{appendix:ellipsoid}

It is obvious that $\bQ^{(c)}=\bH^{(c)}(\bH^{(c)})^\hh$ is positive semidefinite. Moreover, $\bH^{(c)}$ being full row rank implies that $\bQ^{(c)}$ is positive definite. Therefore, \eqref{eq:power_reform} is a positive definite quadratic form of ${\bX}^{(c)}$, which defines an ellipsoid naturally. 

Next, we prove that the ellipsoid is an embedded submanifold of ambient Euclidean space $\mathbb{C}^{KN_r\times Kd_k}$ using \cite[Proposition 3.3.3]{AbsilOpt}. Consider the ellipsoid
\begin{equation}
\mathcal{M}_c = \left\{{\bX}^{(c)}:\operatorname{tr}\left(({\bX}^{(c)})^\hh \bQ^{(c)}{\bX}^{(c)}\right)=P_c\right\}
\end{equation}
Define the smooth mapping $F:\mathbb{C}^{KN_r\times Kd_k}\rightarrow\mathbb R$ by
\begin{equation}
F({\bX}^{(c)})=\operatorname{tr}\left(\bQ^{(c)}{\bX}^{(c)}{\bX}^{(c)\hh}\right)-P_c.
\end{equation}
Then $\mathcal{M}_c = F^{-1}(0)$, where $\mathbb{R}$ has dimension $d = 1$.

To apply \cite[Proposition 3.3.3]{AbsilOpt}, $0$ must be a regular value of $F$, i.e., the differential $\mathrm{D}F({\bX}^{(c)}): T_{{\bX}^{(c)}} \mathbb{C}^{n \times m} \to \mathbb{R}$ has rank $1$ for all $\bX^{(c)} \in \mathcal{M}_c$. The differential is
\begin{equation}
\begin{aligned}
  \mathrm{D}F({\bX}^{(c)})[\bZ] &= 2 \mathrm{Re}\left( \operatorname{tr}\left( \bQ^{(c)} {\bX}^{(c)} \bZ^H \right) \right)\\
  & = 2 \langle\bQ^{(c)} {\bX}^{(c)}, \bZ \rangle_F.
\end{aligned}
\end{equation}
The corresponding Euclidean gradient is $\nabla F({\bX}^{(c)}) = 2\bQ^{(c)} {\bX}^{(c)}$, and
\begin{equation}
\| \nabla F({\bX}^{(c)}) \|_F = 2 \|\bQ^{(c)} {\bX}^{(c)} \|_F.
\end{equation}
Since $\operatorname{tr}\left(\bQ^{(c)}{\bX}^{(c)}{\bX}^{(c)\hh}\right)=P_c>0$, we must have $\bQ^{(c)} {\bX}^{(c)}\neq0$. Therefore, 
\begin{equation}
\| \nabla F({\bX}^{(c)}) \|_F = 2 \|\bQ^{(c)} {\bX}^{(c)} \|_F>0
\end{equation}
$\forall\bX^{(c)} \in \mathcal{M}_c$. Thus, $\nabla F({\bX}^{(c)}) \neq \mathbf{0}$, so $\mathrm{D}F({\bX}^{(c)})$ is surjective (rank $1$). Finally, by \cite[Proposition 3.3.3]{AbsilOpt}, $\mathcal{M}_c$ is an embedded submanifold of $\mathbb{C}^{KN_r\times Kd_k}$.

\bibliographystyle{IEEEtran}
\bibliography{IEEEabrv, Ref}                       

\end{document}